\documentclass[12pt]{article}
\usepackage[english]{babel}
\usepackage{graphicx}
\usepackage{framed}
\usepackage{color}
\usepackage[normalem]{ulem}
\usepackage{amsmath}
\usepackage{amsthm}
\usepackage{amssymb}
\usepackage{amsfonts}
\usepackage{enumerate}
\usepackage[utf8]{inputenc}
\usepackage[top=1 in,bottom=1in, left=1 in, right=1 in]{geometry}
\usepackage{hyperref}
\usepackage{natbib}
\usepackage{multirow}
\usepackage{caption}
\usepackage{subcaption}

\theoremstyle{definition}
\newtheorem{theorem}{Theorem}

\newtheorem{assumption}{Assumption}
\newtheorem{lemma}{Lemma}

\title{CoVaR under Asymptotic Independence}
\author{Zhaowen Wang\\School of Statistics and Data Science,\\Shanghai University of
International Business and Economics,\\Yutao Liu \\School of Statistics and Mathematics, \\Central University of Finance and Economics\\and \\Deyuan Li \\Department of Statistics and Data Science, School of Management,\\ Fudan University}

\date{}
\begin{document}\maketitle
\begin{abstract}
  Conditional value-at-risk (CoVaR) is one of the most important measures of systemic risk. It is defined as the high quantile  conditional on a related
variable being extreme, widely used in the field of quantitative risk management. In this work, we develop a semi-parametric methodology to estimate CoVaR for asymptotically independent pairs within the framework of bivariate extreme value theory. We use parametric modelling of the bivariate extremal structure to address data sparsity in the joint tail regions and prove consistency and asymptotic normality of the proposed estimator. The robust performance of the estimator is illustrated via simulation studies. Its application to the US stock returns data produces insightful dynamic CoVaR forecasts.
\end{abstract}
\section{Introduction}
\label{sec:intro}

The global financial crisis of 2007-2009 highlighted the crucial role systemic risk plays in ensuring stability of financial markets. Unlike idiosyncratic failures, systemic risk represents the danger that distress in a single institution can propagate through complex, interconnected networks, leading to a total collapse of market liquidity and function. The
interconnections among financial institutions highlight the need to better
understand how the risk of one financial institution
affects the risk of another. For more about financial systemic
risk, please refer to \cite{hansen2013challenges}, \cite{engle2018systemic} and \cite{jackson2021systemic}. 

An accurate measurement of systemic risk enables the regulators to implement comprehensive strategies to safeguard the entire financial system. It also provides institutions with the means to scrutinize their susceptibility to market fluctuations. Several measures of systemic risk have been proposed in the literature, including CATFIN designed by \cite{allen2012does}, SRISK developed by \cite{brownlees2017srisk}, the marginal expected shortfall introduced by \cite{acharya2017measuring}, a relative
risk measure constructed by \cite{he2019statistical} and so on. 

One popular measure of systemic risk is the conditional value-at-risk (CoVaR), proposed by \cite{adrian2011covar}. For two random loss variables $X$ and $Y$ and level $p$, $\operatorname{CoVaR}^=_{Y \mid X}(p)$ is defined as$$\mathbb{P}\left(Y \geq \operatorname{CoVaR}^=_{Y \mid X}(p) \mid X=\operatorname{VaR}_X(p)\right)=p, \quad p \in(0,1),$$where $\operatorname{VaR}_X(p)$ is the value-at-risk (VaR) for random variable $X$ at level $1-p$, given by
$$
\operatorname{VaR}_X(p)=\inf _x\{\mathbb{P}(X>x) \leq p\}.$$Typical values of interest for the risk level $p$ are $1 \%$ and $5\%$. While VaR measures the risk of individual institutions in isolation, \cite{adrian2011covar} show that CoVaR well captures the cross-sectional tail-dependency between the whole financial
system and a particular institution, and could predict
the 2007–2009 crisis. 

\cite{girardi2013systemic} generalize the definition of CoVaR by assuming that the conditioning distress event refers to exceedence $\{X \geq \operatorname{VaR}_X(p)\}$ as opposed to equality $\{X=\operatorname{VaR}_X(p)\}$. The definition of CoVaR is modified as
\begin{equation}\mathbb{P}\left(Y \geq \operatorname{CoVaR}_{Y \mid X}(p) \mid X \geq \operatorname{VaR}_X(p)\right)=p, \quad p \in(0,1).
\label{covar}
\end{equation}This change allows to consider more severe distress events that are farther in the tail (above its VaR) and to backtest CoVaR estimates using the standard approaches for backtesting VaR. \cite{mainik2014dependence} also show this definition leads to CoVaR being dependence consistent.  

Quantile regression techniques are often employed in estimating $\operatorname{CoVaR}^=_{Y \mid X}(p)$, see, e.g., \cite{adrian2011covar} and \cite{leng2024asymptotics}. \cite{huang2024monte} also develop a Monte Carlo algorithm to estimate CoVaR under this definition. On the other hand, the estimation methods for $\operatorname{CoVaR}_{Y \mid X}(p)$  are often fully parametric, including the GARCH model \citep{girardi2013systemic},  dynamic copula model \citep{oh2018time} and \citep{bianchi2023non}. \cite{nolde2020conditional} and \cite{nolde2022extreme} develop semiparametrical approaches based on extreme value theory.

In the present paper, we interpret $X$ and $Y$ as losses of a financial institution and a system proxy. We note that the approach in \cite{nolde2022extreme} is based on the assumption that $X$ and $Y$ are asymptotic dependent. If $X$ and $Y$ are asymptotic independent, this approach fails. Although most research articles in bivariate extreme value framework deal with asymptotic dependence, there is increasing evidence that weaker dependence actually exists in bivariate tail region in many applications, for example, significant wave height \citep{wadsworth2012dependence}, spatial precipitation \citep{le2018dependence} and daily stock prices \citep{lehtomaa2020asymptotic}. Asymptotic independence is therefore the more appropriate model for such applications. In the field of quantitative risk management, there is also growing attention on risk measures for asymptotically independent pairs \citep[see][etc]{kulik2015heavy,das2018risk,cai2020estimation,sun2022extreme,wang2023tail}. 
In this paper, we will study the estimation of $\operatorname{CoVaR}_{Y \mid X}(p)$ under asymptotic independence. We refer to \cite{ledford1996statistics} for the concepts of asymptotic independence and asymptotic dependence.

To incorporate asymptotic independence in statistical modelling, we adopt the condition on the joint exceedance probabilities by
\cite{ledford1997modelling}, which assumes that there exists $0<\eta\leq1$, such as for all $x, y \in[0, \infty)$, for following limit exist
\begin{equation}
\lim _{p\to 0} p^{-\frac{1}{\eta}} \mathbb{P}\left(F_1(X) \geq 1-p x, F_2(Y) \geq 1-p y\right)=:c(x, y)\in[0, \infty),
\label{1}
\end{equation}
where $F_1$ and $F_2$ are the distribution functions of random variables X and Y, respectively. One important property of function $c$ is that it is homogeneous of order $1/\eta$, i.e., $c(ax,ay) =a^{1/\eta} c(x,y)$ for $a>0$. The coefficient of tail dependence $\eta$ describes the strength of extremal dependence in the bivariate tail. If $\eta=1$, we say that $X$ and $Y$ are asymptotically dependent, and $c(x, y)$ is the tail dependence function. If $0<\eta<1$, we say that $X$ and $Y$ are asymptotically independent. Moreover, if $1/2<\eta<1$, $X$ and $Y$ are called asymptotically independent but positively associated; if $0<\eta<1/2$, $X$ and $Y$ are called asymptotically independent but negatively associated. When $X$ and $Y$ are independent, then $\eta=1/2.$ For more details on the interpretation of $\eta$ see \cite{ledford1996statistics} and \cite{ledford1997modelling}.

It is the goal of this paper to estimate  $\operatorname{CoVaR}_{Y \mid X}(p)$ under asymptotic independence. However, we exclude the nagative association and require that $1/2<\eta<1$ for the sake of identification.

The rest of the paper is organized as follows. Section \ref{sec2} proposes an estimator for $\operatorname{CoVaR}_{Y \mid X}(p)$ and studies its consistency and asymptotic nomality. The performance of our proposed estimator is illustrated by a simulation study in Section \ref{sim}, and a real application to S\&P 500 is given in Section \ref{sec4}. The proofs of the main theorems are provided in Section \ref{pf}. 


\section{Main results}
\label{sec2}\subsection{Probabilistic framework}
Recall the adjustment factor $\eta_{p}$ introduced in \cite{nolde2022extreme}
\begin{equation}
\eta_{p}=\frac{\mathbb{P}\left(Y \geq \operatorname{CoVaR}_{Y \mid X}(p)\right)}{\mathbb{P}\left(Y \geq \operatorname{CoVaR}_{Y \mid X}(p) \mid X \geq \operatorname{VaR}_X(p)\right)}, \quad p \in(0,1) .
\label{def}\end{equation}Hence CoVaR is related to a quantile of the unconditional distribution of $Y$ via
\begin{equation}
\operatorname{CoVaR}_{Y \mid X}(p)=\operatorname{VaR}_Y\left(p \eta_{p}\right)
\label{equiv}
\end{equation}with the adjusted quantile level $p \eta_{p}$. If $X$ and $Y$ are independent, then CoVaR coincides with the VaR of $Y$ at the same level and $\eta_{p}=1$. When $1/2<\eta\leq1$, we have $\eta_{p}<1$ and CoVaR is equal to $\mathrm{VaR}$ at a higher confidence level determined by $\eta_{p}$.

Going back to the definition of CoVaR in (\ref{covar}) and using (\ref{equiv}), we have
$$\mathbb{P}\left(X\geq\operatorname{VaR}_X(p), Y\geq\operatorname{VaR}_Y\left(p \eta_{p}\right)\right)=\mathbb{P}\left(1-F_1(X)\leq p, 1-F_2(Y)\leq p\eta_{p}\right)=p^{2}, 
$$Denote by $Q$ the joint distribution function of $\left(1-F_1(X), 1-F_2(Y)\right)$, and we can write\begin{equation}
Q\left(p, p\eta_{p}\right)=p^{2}.\label{qequa}\end{equation}
This suggests a possibility of approximating the true adjustment factor $\eta_{p}$ with an asymptotically determined approximation, denoted $\eta_{p}^*$, which is defined implicitly via
\begin{equation}
c\left(1, \eta_{p}^*\right)=p^{2-1/\eta}.\label{equa}
\end{equation}

Note that if $\eta=1$, then relation (\ref{equa}) is exactly the same as the case in \cite{nolde2022extreme}. In situations where function $c(x,y)$ provides a good approximation of the dependence structure in the tail region, we expect $\eta_{p}$ and $\eta_{p}^*$ to be close. We prove this formally in Section \ref{pf}. Theorem 1 in \cite{li2025properties} also presents similar results. Notice that a larger value of $\eta$, i.e. stronger tail dependence, will lead to a smaller value of $\eta_{p}^*$. 

As the tail dependence function is monotonically non-decreasing in each coordinate, it follows that $c(1, s)$ is an increasing function of $s$ from zero to the value of $c(1,1)$ for values of $s$ from zero to one. Hence, a unique solution $\eta_{p}^*$ to equation (\ref{equa}) exists provided that $0<p^{2-1/\eta}<c(1,1)$. In applications, $p$ will typically be taken to be small and so generally it will be possible to find the solution as long as dependence is not too close to the nagative association case. 

Thus $\operatorname{CoVaR}_{Y \mid X}(p)$ can by approximated by the quantile of the unconditional distribution $\operatorname{VaR}_Y\left(p \eta_{p}^*\right)$, and regular variation of $1-F_2$ can be used to give further approximation:
$$
\operatorname{CoVaR}_{Y \mid X}(p) \approx \operatorname{VaR}_Y\left(p \eta_{p}^*\right) \approx\left(\eta_{p}^*\right)^{-\gamma} \operatorname{VaR}_Y(p), 
$$where $\gamma$ is the extreme value index of the distribution of $Y$. The above expression will be used as a basis for constructing an asymptotically motivated estimator of CoVaR.

\subsection{Estimation}Suppose that $\left(X_1, Y_1\right), \ldots,\left(X_n, Y_n\right)$ are independent copies of $(X, Y)$. Define three intermediate sequences $k_1, k_2, k_3$, satisfying $k_i\to\infty$ and $k_i/n\to0,$ for $ i=1,2,3.$
Using ideas outlined above, we propose the following estimator of $\operatorname{CoVaR}_{Y \mid X}(p)$ for small values of $p$ :
\begin{equation}
\widehat{\operatorname{CoVaR}}_{Y \mid X}(p)=\left(\widehat{\eta}_p^*\right)^{-\widehat{\gamma}} \widehat{\operatorname{VaR}}_Y(p) .
\label{cov}\end{equation}

In order to estimate $\operatorname{CoVaR}_{Y \mid X}(p)$, we have to estimate ${\eta}_p^*, \gamma$ and $\operatorname{VaR}_Y(p)$. Extreme value index $\gamma$ of the distribution of $Y$ assumed to have a regularly varying tail can be estimated using the Hill estimator by \cite{hill1975simple}:
$$
\widehat{\gamma}=\frac{1}{k_1} \sum_{i=1}^{k_1} \log Y_{n, n-i+1}-\log Y_{n, n-k_1},
$$where $Y_{n,1}\leq Y_{n,2}\leq\cdots\leq Y_{n,n}$ are the order statistics of $\{Y_{1}, Y_{2},\cdots,Y_{n}\}$.

The estimator of $\operatorname{VaR}_Y(p)$ can be computed using the extrapolation method, see Section 4.3 in \cite{de2007extreme}:
\begin{equation}
\widehat{\operatorname{VaR}}_Y(p)=Y_{n, n-k_2}\left(\frac{k_2}{n p}\right)^{\widehat{\gamma}}.\label{var}
\end{equation}

Once we estimate $\eta_{p}^*$, we can build the estimator of  $\operatorname{CoVaR}$ by (\ref{cov}). Finding $\widehat{\eta}_p^*$ in (\ref{cov}) requires estimates of both the coefficient of tail dependence $\eta$ and the function $c(x,y)$. 
For $c(x,y)$, in light of data sparsity in the tail region, we impose a parametric assumption on its form so as to facilitate the computation of an estimate of $\eta_{p}$. Similar to the settings in \cite{lalancette2021rank}, we assume the function $c$ belongs to some parametric family $\{c(\cdot, \cdot ; \boldsymbol{\theta}): \boldsymbol{\theta} \in \Theta\}$, where $\Theta \subset \mathbb{R}^s(s \geq 1)$ is the parameter space.  We adopt the method-of-moments (M-estimator) proposed in \cite{lalancette2021rank} to estimate $\boldsymbol{\theta}$. It is an extention of \cite{einmahl2012m} and asymptotic normality hold under weak conditions.

Denote $\widehat{Q}_n(x, y)$ as an nonparametric estimator of $Q(x,y)$, defined by
$$
\widehat{Q}_n(x, y):=\frac{1}{n} \sum_{i=1}^n \mathbf{1}\left\{n \widehat{F}_1\left(X_i\right) \geq n+1-\lfloor k_3 x\rfloor, n \widehat{F}_2\left(Y_i\right) \geq n+1-\lfloor k_3 y\rfloor\right\},$$ where  $\widehat{F}_1$ and $\widehat{F}_2$ are the empirical marginal distributions of $F_1$ and $F_2$, respectively and $\lfloor a \rfloor$ is the largest integer less than or equal to $a$. Let $g=\left(g_1, \ldots, g_s\right)^T:[0,1]^2 \rightarrow \mathbb{R}^s$ be a vector of integrable functions and $\boldsymbol{\theta}_0$ denote the true value of parameter $\boldsymbol{\theta}$. Define
$$
\Phi_{n}(\boldsymbol{\theta}, \zeta)=\zeta\iint_{[0,1]^2} g(x, y) c(x, y ; \boldsymbol{\theta}) d x d y -\iint_{[0,1]^2} g(x, y) \widehat{Q}_n(x, y) d x d y,
$$and let \begin{equation}
  (\widehat{\boldsymbol{\theta}}, \widehat{\zeta}):=\text{arg min}_{\boldsymbol{\theta}\in\Theta,\zeta>0}\left\|\Phi_n(\boldsymbol{\theta}, \zeta)\right\|,\label{para}\end{equation}
where $\|\cdot\|$ is the Euclidean norm and $\zeta$ is an additional scale parameter. The choice of function $g$ is discussed further in Section \ref{sim}. 

Recall that $\eta$ can be recovered from the function $c$ since the latter is homogeneous of order $1/\eta$. Therefore, inside the assumed parametric model, $\eta$ can be represented as a function $\eta=\eta(\boldsymbol{\theta})$ and the estimator $\widehat{\eta}$ can be represented as $\widehat{\eta}=\eta(\widehat{\boldsymbol{\theta}})$.Once we have $\widehat{\boldsymbol{\theta}}$ and $\widehat{\eta},$ the estimator $\widehat{\eta}^*_p$ in  (\ref{cov}) can obtained by solving 
\begin{equation}
c\left(1, \widehat{\eta}^*_p ;\widehat{\boldsymbol{\theta}}\right)=p^{2-1/\widehat{\eta}}.
\label{hat}\end{equation}\subsection{Consistency}\label{consis}In this subsection we study the consistency of the proposed CoVaR estimator in (\ref{cov}). Denote the quantile function $U_2=\left(1 / 1-F_2\right)^{\leftarrow}$, where $^{\leftarrow}$ is the left-continuous inverse. Then clearly $\operatorname{VaR}_Y(p)=U_2(1 / p)$.  Define
$$
\widetilde{\Phi}(\boldsymbol{\theta}, \zeta)=\zeta\iint_{[0,1]^2} g(x, y) c(x, y ; \boldsymbol{\theta}) d x d y -\iint_{[0,1]^2} g(x, y) c(x, y ; \boldsymbol{\theta}_0)  d x d y.
$$

Firstly, we present three assumptions that are necessary for consistency of $\widehat{\operatorname{VaR}}_Y(p)$ in (\ref{var}); see, e.g., Theorem 4.3.8 in \cite{de2007extreme}. Assumption \ref{asm1} is a second order condition for 
the distribution of $Y$, which is commonly assumed in extreme value theory. Assumption \ref{asm2} is for the two intermediate sequences $k_j=k_j(n), j=1,2$, used in the estimator of VaR in (\ref{var}), and Assumption \ref{asm3} resticts the order of the probability level $p=p(n)$.

\begin{assumption}\label{asm1}
  There exist a constant $\rho<0$ and an eventually positive or negative function $A(t)$ such that as $t \rightarrow \infty, A(t) \rightarrow 0$ and for all $x>0$,
    $$
    \lim _{t \rightarrow \infty} \frac{\frac{U_2(t x)}{U_2(t)}-x^\gamma}{A(t)}=x^\gamma \frac{x^\rho-1}{\rho} .
    $$
\end{assumption}\begin{assumption}\label{asm2}As $n \rightarrow \infty, \sqrt{k_j} A\left(n / k_j\right)\rightarrow \lambda_j \in \mathbb{R},$ for $j=1,2$. \end{assumption} 
\begin{assumption}\label{asm3}
    As $n \rightarrow \infty, k_2/(n p)\rightarrow \infty $, and$ \sqrt{k_1}/\log \left(k_2 / (n p)\right)\rightarrow \infty.
    $\end{assumption}
    
    Next, we give Assumptions \ref{asm4} and \ref{asm4+}, which aim at controlling the speed of convergence to the limit in (\ref{1}).  Assumption \ref{asm4} also ensures that the approximation in (\ref{equa}) is reasonable. Assumption \ref{asm5} is taken from \cite{lalancette2021rank} to show the asymptotic normality of M-estimator $\widehat{\boldsymbol{\theta}}_n$.\begin{assumption}\label{asm4} 
    For $x, y\in \mathcal{S}^{+}=\left\{(x, y) \in[0, \infty)^2: x^2+y^2=1\right\}$, $$p^{-1/\eta}Q(px, py)=c(x, y)+O\left(q_1(p)\right), \quad x, y \in[0, \infty)$$holds uniformly, where  $q_1(p)=o(p^{2-1/\eta})$ as $p\to0$.
    \end{assumption}\begin{assumption}\label{asm4+}
      As $n \rightarrow \infty, m=m_n:=n(k_3/n)^{1/\eta} \rightarrow \infty $, and$ \sqrt{m}q_1\left(k_3 / n \right)\rightarrow 0.
      $\end{assumption}
\begin{assumption}\label{asm5}
 The function $\widetilde{\Phi}$ has a unique, well separated zero at $\left(\boldsymbol{\theta}_0, 1\right)$ and is differentiable at that point with Jacobian $J:=J_{\widetilde{\Phi}}\left(\boldsymbol{\theta}_0, 1\right)$ of full rank $s+1$.
\end{assumption}Lastly, we further impose one condition on the partial derivative of function $c(x,y; \boldsymbol{\theta})$. Define $c_y(x, y ; \boldsymbol{\theta}):=\partial c(x, y ; \boldsymbol{\theta}) / \partial y$. Assumption \ref{asm6} below ensures that $\eta_{p}^* \rightarrow 0$ with the same speed as $p\rightarrow 0$ for all functions $c(\cdot, \cdot ; \boldsymbol{\theta})$ in the parametric family.
\begin{assumption}\label{asm6}
For all $\boldsymbol{\theta} \in \Theta$, the partial derivative $c_y(x, y ; \boldsymbol{\theta})$ is continuous with respect to $y$ in the neighborhood of $(1,0 ; \boldsymbol{\theta})$ and $c_y(1,0 ; \boldsymbol{\theta})>0$.
\end{assumption}
    
Consider the CoVaR estimator defined in (\ref{cov}). Here, $\widehat{\gamma}$ is estimated by the Hill estimator; $\widehat{\operatorname{VaR}}_Y(p)$ is estimated using (\ref{var}) and $\widehat{\eta}_p^*$ is estimated with (\ref{hat}). The following theorem shows consistency of the CoVaR estimator. The proof is given in Section \ref{pf}.
\begin{theorem}\label{thm1}
  Under Assumptions \ref{asm1} to \ref{asm6}, as $n \rightarrow \infty$,$$\frac{\widehat{\operatorname{CoVaR}}_{Y \mid X}(p)}{\operatorname{CoVaR}_{Y \mid X}(p)} \stackrel{\mathbb{P}}{\rightarrow} 1.$$
\end{theorem}
\subsection{Asymptotic normality}In this subsection we present the asymptotic normality of $\widehat{\operatorname{CoVaR}}_{Y \mid X}(p)$ in (\ref{cov}) under some further assumptions. 
\begin{assumption}\label{asm7}For each given $\boldsymbol{\theta} \in \Theta$, there exists a function $\breve{\rho}(\boldsymbol{\theta})>0$, continuous at $\boldsymbol{\theta}=\boldsymbol{\theta}_0$, such that as $y \rightarrow 0$,
  $$
  \frac{c(1, y ; \boldsymbol{\theta})}{y}-c_y(1,0 ; \boldsymbol{\theta})=O\left(y^{\check{\rho}(\boldsymbol{\theta})}\right),
  $$and 
$c_y(1,0 ; \boldsymbol{\theta})$ is 1-Lipschitz continuous in the neighborhood of $\boldsymbol{\theta}_0.$ In addition, there exists a constant $\varepsilon>0$ such that $\sqrt{m}p^{\breve{\rho}(\boldsymbol{\theta}_0)(2-1/\eta)-\varepsilon}\to 0$, as $n\to\infty$.
  \end{assumption}

\begin{assumption}\label{asm8}
  As $n\to\infty, \frac{\sqrt{k_1}}{\log \left(k_2 / n p\right)}\max\{\frac{1}{\sqrt{k_2}}, \frac{1}{\sqrt{m}},\frac{q_1(p)}{p^{2-1/\eta}}\} \to 0.$
\end{assumption}
We remark that all of the above assumptions are compatible. For example, one can choose $k_1 =k_2=k_3=n^{\kappa}$ with $1-\eta<\kappa<1$ and $p=n^{\zeta}$ with $\zeta<\frac{-1/2-(\kappa-1)/(2\eta)}{\breve{\rho}(\boldsymbol{\theta}_0)(2-1/\eta)-\varepsilon}.$ The following theorem shows the asymptotic normality of the CoVaR estimator in (\ref{cov}). The proof is given in Section \ref{pf}.
 \begin{theorem}\label{thm2}
  Under Assumptions \ref{asm1} to \ref{asm8}, as $n \rightarrow \infty$,$$\frac{\sqrt{k_1}}{\log \left(k_2 / n p\right)}\left(\frac{\widehat{\operatorname{CoVaR}}_{Y \mid X}(p)}{\operatorname{CoVaR}_{Y \mid X}(p)}-1\right)\stackrel{d}{\rightarrow}\mathcal{N}(\frac{\lambda}{1-\rho},\gamma^2).$$
\end{theorem}
\subsection{Extensions}
In this subsection, we discuss some extensions of CoVaR definitions and propose corresponding estimation methods. \subsubsection{Estimation of CoVaR\texorpdfstring{$_{Y \mid X}(p, q)$}{_{Y \mid X}(p, q)}}Consider allowing the quantile levels of $X$ and $Y$ in CoVaR to be different. For two quantile levels $p, q \in(0,1)$, define $\operatorname{CoVaR}_{Y \mid X}(p, q)$by $$
\mathbb{P}\left(Y \geq \operatorname{CoVaR}_{Y \mid X}(p,q) \mid X \geq \operatorname{VaR}_X(p)\right)=q.
$$Several studies including \cite{huang2024monte} and \cite{mainik2014dependence} investigate the estimation and properties of $\operatorname{CoVaR}_{Y \mid X}(p, q)$. The adjustment factor $\eta_{p,q}$ for $\operatorname{CoVaR}_{Y \mid X}(p, q)$ can be defined as 
$$\eta_{p,q}=\frac{\mathbb{P}\left(Y \geq \operatorname{CoVaR}_{Y \mid X}(p,q)\right)}{\mathbb{P}\left(Y \geq \operatorname{CoVaR}_{Y \mid X}(p,q) \mid X \geq \operatorname{VaR}_X(p)\right)}.$$
Assume that $q=Cp$, where $C>0$ is a constant, we can get the approximate adjustment factor $\eta_{p,q}^*$ by solving the equation \begin{equation*}c\left(1, C\eta_{p,q}^*\right)=Cp^{2-\frac{1}{\eta}}\end{equation*}and the estimator of $\operatorname{CoVaR}_{Y \mid X}(p,q)$ is constructed via 
\begin{equation*}
\widehat{\operatorname{CoVaR}}_{Y \mid X}(p)=\left(\widehat{\eta}_{p,q}^*\right)^{-\widehat{\gamma}} \widehat{\operatorname{VaR}}_Y(p),\end{equation*} where $\widehat{\eta}_{p,q}^*$ is obtained by solving 
\begin{equation*}
c\left(1, C\widehat{\eta}^*_{p,q} ;\widehat{\boldsymbol{\theta}}\right)=Cp^{2-1/\widehat{\eta}}.
\end{equation*} 

The consistency and asymptotic normality of $\widehat{\operatorname{CoVaR}}_{Y \mid X}(p,q)$ can be established in the same manner as in Theorems \ref{thm1} and \ref{thm2}.
\subsubsection{Estimation of CoVaR\texorpdfstring{$^=_{Y \mid X}(p)$}{^=_{Y \mid X}(p)} and CoVaR\texorpdfstring{$^=_{Y \mid X}(p, q)$}{^=_{Y \mid X}(p, q)}}Consider the original definition of CoVaR proposed by \cite{adrian2011covar}, i.e., $$\mathbb{P}\left(Y \geq \operatorname{CoVaR}^=_{Y \mid X}(p) \mid X=\operatorname{VaR}_X(p)\right)=p,$$ where the adjustment factor $\eta^=_p$ is defined by \begin{equation*}
\eta^=_p=\frac{\mathbb{P}\left(Y \geq \operatorname{CoVaR}^=_{Y \mid X}(p)\right)}{\mathbb{P}\left(Y \geq \operatorname{CoVaR}^=_{Y \mid X}(p) \mid X =\operatorname{VaR}_X(p)\right)}, \quad p \in(0,1) .
\end{equation*} 
Define the partial derivative $c_x(x, y ; \boldsymbol{\theta}):=\partial c(x, y ; \boldsymbol{\theta}) / \partial x$. We follows Theorem 2 and 3 in \cite{li2025properties}, and define the approximate adjustment factor $\eta^{*=}_p$ by the equation \begin{equation*}
c_x\left(1, \eta^{*=}_{p}; \boldsymbol{\theta}\right)=p^{2-1/\eta}.
\end{equation*}The estimators $\widehat{\eta}_{p}^{*=}$ and $\widehat{\operatorname{CoVaR}^=}_{Y \mid X}(p)$ can be similarly obtained.

For $p, q \in(0,1)$, define  $\operatorname{CoVaR}^=_{Y \mid X}(p, q)$ by, $$\mathbb{P}\left(Y \geq \operatorname{CoVaR}^=_{Y \mid X}(p,q) \mid X=\operatorname{VaR}_X(q)\right)=p.$$ Thus the adjustment factor $\eta^=_{p,q}$ for $\operatorname{CoVaR}^=_{Y \mid X}(p, q)$ can be defined as 
$$\eta^=_{p,q}=\frac{\mathbb{P}\left(Y \geq \operatorname{CoVaR}^=_{Y \mid X}(p,q)\right)}{\mathbb{P}\left(Y \geq \operatorname{CoVaR}^=_{Y \mid X}(p,q) \mid X= \operatorname{VaR}_X(p)\right)}.$$
Assume that $q=Cp$, where $C>0$ is a constant, by Theorems 2 and 3 in \cite{li2025properties}, we get the approximate adjustment factor $\eta_{p,q}^{*=}$ by solving the equation \begin{equation*}c_x\left(1, C\eta_{p,q}^{*=}\right)=Cp^{2-\frac{1}{\eta}}. \end{equation*}The estimators $\widehat{\eta}_{p,q}^{*=}$ and $\widehat{\operatorname{CoVaR}^=}_{Y \mid X}(p,q)$ can be similarly obtained. 

\section{Simulation}\label{sim}
We assess the finite sample properties of the proposed $\operatorname{CoVaR}$ estimator in (\ref{cov}) in comparison with the naive estimator based on the order statistic. The simulation is based on 1000 repetitions and samples of size $n=5000$ at risk level $p=5 \%$. To make our estimation procedure automatic for the purpose of simulation studies, we set $k=k_1=k_1=k_3$ and take two values $k=1000$ and $k=1500$, respectively. The samples are simulated from the following two models:

Model 1: Let $Z_1, Z_2$, and $Z_3$ be independent Pareto random variables with parameters $\theta_1, \theta_2$, and $\theta_1$, respectively. Here, a Pareto distribution with parameter $\theta>0$ means that the probability density function is $f(x)=\theta^{-1}x^{-1 /\theta-1}$ for $x>1$. Define
$$
(X, Y)=B\left(Z_1, Z_3\right)+(1-B)\left(Z_2, Z_2\right),
$$
where $B$ is a $\text{Bernoulli}(1/2)$ random variable independent of $Z_i$'s. By \cite{cai2020estimation}, we have $$c(x, y)=2^{\theta_1 /\theta_2-1}(x \wedge y)^{\theta_1 /\theta_2},\quad\eta_{p}^*=p^{\frac{2\theta_2}{\theta_1}-1}2^{\frac{\theta_2}{\theta_1}-1}, \quad\eta=\theta_2/\theta_1, \quad q_1(p)=p^{\frac{\theta_1}{\theta_2}-1}.$$ We note that Model 1 satisfies Assumptions \ref{asm4} as long as $\theta_1/\theta_2>3/2$. Meanwhile, to ensure $1/2<\eta<1$, we need $3/2<\theta_1/\theta_2<2$.

Model 2: Let $(X, Y)$ follow the inverted Hüsler-Reiss distribution with unit Fréchet margins. Here, \begin{equation}Q(x, y)=\exp(-\ell(-\log x,-\log y)), \quad x, y \in[0, 1],
  \label{inv}
  \end{equation}holds with $$
  \ell(x, y)=x \Phi\left(\lambda+\frac{\log x-\log y}{2 \lambda}\right)+y \Phi\left(\lambda+\frac{\log y-\log x}{2 \lambda}\right)
  $$
  where $\Phi$ is the standard normal distribution function and $\lambda \in[0, \infty]$ parametrizes between perfect independence $(\lambda=\infty)$ and dependence $(\lambda=0)$. Let $\theta=\Phi(\lambda)$ and the parameter space  $\Theta=(1 / 2,1]$. By \cite{lalancette2021rank}, we have $$c(x, y)=(x y)^\theta,\quad\eta_{p}^*=p^{\frac{2-2\theta}{\theta}},\quad \eta=1/2\theta,\quad q_1(p)=1/\log(1/p).
$$ We note that in this model $q_1(p)=o(p^{2-1/\eta})$ does not hold and Assumption \ref{asm4} is not satisfied.

As for the weight function $g$, we consider 
\begin{equation}
  g(x, y):=\left(\mathbf{1}\left\{(x, y) \in I_1\right\} / a_{1, \theta_{\mathrm{REF}}}, \ldots, \mathbf{1}\left\{(x, y) \in I_5\right\} / a_{5, \theta_{\mathrm{REF}}}\right).
\label{wt}
\end{equation} In Model 1, we let $I_1=[0,1]^2$, $I_2=[0,0.8]\times[0,1], I_3=[0,1]\times[0,0.5], I_4=[0,0.5] \times[0,0.3]$ and $I_5=[0,0.5] \times[0,0.5]$. In Model 2, we let $I_1=[0,1]^2$, $I_2=[0,2]^2, I_3=[1 / 2,3 / 2]^2, I_4=[0,1] \times[0,3]$ and $I_5=[0,3] \times[0,1]$, which follows \cite{lalancette2021rank}. Here, $a_{j, \theta_{\mathrm{REF}}}:=\int_{I_j} c_{\theta_{\mathrm{REF}}}(x,y) d xdy$ and $\theta_{\mathrm{REF}}$ is simply a reference point in the parameter space that ensures that all components of $g$ have comparable magnitude. In the two models above, the reference points are $(0.6,0.6)$ and $0.6$, respectively. 

The settings of the simulation studies are summarized in Table \ref{tbsim}. The true value of $\operatorname{CoVaR}_{Y \mid X}(p)$ is computed by solving the adjustment factor $\eta_{p}$ in (\ref{qequa}) and then plugging in (\ref{equiv}). For each model, we select three different parameters. The mean and standard deviation of the proposed CoVaR estimators for the two intermediate levels, as well as the naive estimator, are also reported in Table \ref{tbsim}. \begin{table}[]\centering\caption{True values of $\operatorname{CoVaR}_{Y \mid X}(p)$ and means of the estimated values with corresponding standard deviation given in the brackets and the naive estimator at level $p = 0.05$ .}\label{tbsim}
  \begin{tabular}{cccccl}
    \hline
                          & $\theta$     & True value & $k=1000$     & $k=1500$    & Naive        \\ \hline
    \multirow{3}{*}{Model 1} & (0.85, 0.45) & 11.17      & 13.76(1.88)  & 14.30(1.42) & 18.26(3.57)  \\
                          & (0.80, 0.42)  & 9.52       & 11.64(1.50)   & 12.08(1.12)  & 15.22(2.80)  \\
                          & (0.75, 0.40)  & 8.55       & 10.20(1.23)   & 10.55(0.92)  & 13.06(2.22)  \\ \hline
    \multirow{3}{*}{Model 2} & 0.91         & 35.54      & 38.42(11.29) & 40.21(8.17) & 40.49(14.46) \\
                          & 0.93         & 30.85      & 32.59(9.34)  & 33.74(6.70) & 35.30(11.74) \\
                          & 0.95         & 26.90      & 28.85(8.16)  & 29.03(5.79) & 31.10(10.54)\\\hline
    \end{tabular}
  \end{table}

  We can make the following conclusions from Table \ref{tbsim}. First, our proposed estimator in (\ref{cov}) is more accurate than the naive estimator in terms of both bias and standard deviation even when the model assumptions are partially violated in Model 2. Second, in both Model 1 and Model 2, while a larger $k$ decreases the variance, a smaller $k$ can decrease the bias, which is a common tradeoff in the studies of extreme value statistics. Third, 
  the standard deviation in Model 1 is significantly smaller than in Model 2, as Model 2 fails to meet all necessary technical conditions. Last but not least, the proposed estimator tends to overestimate the true value, a pattern also observed under asymptotic dependence in \cite{nolde2022extreme}. From the applied perspective, the proposed estimation procedure offers a conservative estimator of systemic risk as measured by CoVaR.
\section{Application}\label{sec4}In this section, we illustrate how the CoVaR estimation methodology can be utilized to produce dynamic CoVaR forecasts using financial time series. 
\subsection{Data description}
In our application, we employ the same dataset as \cite{nolde2022extreme} and consider the daily losses of S\&P 500 index and 15 financial institutions with market capitalizations in excess of 5 billion USD as of the end of June 2007, including Aflac Inc. (AFL), American International Group Inc. (AIG), Allstate Corp. (ALL), Bank Of America Corp. (BAC), Citigroup Inc. (C), Comerica Inc. (CMA), Humana Inc. (HUM), JPMorgan Chase \& Co. (JPM), Lincoln National Corp. (LNC), Progressive Corp. (PGR), USA Education Inc. (SLM), Travelers Companies Inc. (TRV), UnumProvident Corp. (UNM), Wells Fargo \& Co. (WFC) and Washington Mutual Inc. (WM). The S\&P 500 index is used as a system proxy. The sample period is from January 1, 2000 to December 30, 2021, consisting of $n=5535$ daily closing price records for each time series. The daily losses (\%) were calculated as negative log returns.  
\subsection{CoVaR estimation in the dynamic setting}


Let $\left\{X_t^i\right\}_{t \in \mathbb{N}}$ and $\left\{Y_t\right\}_{t \in \mathbb{N}}$ denote time series of daily losses for institution $i$ and system proxy adapted to the filtrations $\mathcal{F}^{X_i}=\left\{\mathcal{F}_t^{X_i}\right\}_{t \in \mathbb{N}}$ and $\mathcal{F}^Y=\left\{\mathcal{F}_t^Y\right\}_{t \in \mathbb{N}},$ respectively, for $i=1,\cdots, 14$. To produce dynamic forecasts, we next define conditional versions of risk measures at time $t$ given information in the series up to time $t-1$. The (conditional)
VaR at confidence level $p\in(0,1)$ for $X_t^i$ given information on the institution's losses up to time $t-1$, denoted $\operatorname{VaR}_{X_t^i}(p)$, is defined as the $\left(1-p\right)$-quantile of the distribution of $X_t^i$ conditional on $\mathcal{F}_{t-1}^{X_i}$ :
$$
\mathbb{P}\left(X_t^i \geq \operatorname{VaR}_{X_t^i}(p)\mid \mathcal{F}_{t-1}^{X_i}\right)=p,
$$
and $\operatorname{CoVaR}_{Y_t\mid X_t^i}(p)$ is defined as the $\left(p\right)$-quantile of the conditional loss distribution given information on losses up to time $t-1$ for both the institution and the system proxy:
$$
\mathbb{P}\left(Y_t \geq \operatorname{CoVaR}_{Y_t\mid X_t^i}(p)\mid X_t^i \geq \operatorname{VaR}_{X_t^i}(p); \mathcal{F}_{t-1}^{X_i}, \mathcal{F}_{t-1}^Y\right)=1-p.
$$

We note that the proposed approach based on extreme value theory is static in nature. Since there is significance evidence of time changing correlation structure in the data, we use a bivariate time series filter with the static estimate for the realized residuals. This is in the spirit of the approach advocated by both \cite{girardi2013systemic} and \cite{nolde2020conditional}, where the authors adopted a bivariate AR(1)-GARCH(1,1) model with the  Dynamic conditional correlation (DCC) specification in \cite{engle2002dynamic}, and use a standardized bivariate skew-$t$ distribution to model innovations. The assumption of no serial correlation has been validated by sample autocorrelation function (acf) plots and Ljung–Box test in \cite{nolde2022extreme}.

The estimation procedure can be summarized in two steps:

\textbf{Step 1} (Bivariate DCC-GARCH model estimation): Assuming the \cite{engle2002dynamic} DCC-GARCH model, the joint losses $\mathbf{X}_t^i=\left(X_t^i, Y_t\right)$ satisfy the following specification:
$$
\begin{aligned}
\mathbf{X}^i_t & =\boldsymbol{\mu}^i_t+\boldsymbol{\epsilon}_t^i, \\
\boldsymbol{\epsilon}_t^i & =(\Sigma_t^i)^{1 / 2} \mathbf{Z}^i_t,
\end{aligned}
$$where $\boldsymbol{\mu}^i_t=\left(\mu_t^i, \mu_t^Y\right)$ is the vector of conditional means and $$\Sigma_t^i=\left( \begin{array}{cc}
  \sigma^2_{t, i}  & \rho_{t,i}\sigma_{t, i}\sigma_{t, Y}  \\\rho_{t,i}\sigma_{t, i}\sigma_{t, Y} & \sigma^2_{t, Y} 
  \end{array} \right)$$ is the conditional covariance matrix of the error term $\boldsymbol{\epsilon}_t^i$. The standardized innovation vectors $\mathbf{Z}_t^i=\left(Z_t^i, Z_t^Y\right)$ are iid with zero mean and identity covariance matrix. Parameters of the process are estimated using R package \textit{bayesDccGarch} and assuming the standardized bivariate skew-$t$ distribution of \cite{bauwens2006multivariate}. The residuals $$\widehat{\mathbf{Z}}_t^i:=\left(\widehat{Z}_t^i, \widehat{Z}_t^Y\right)=(\widehat{\Sigma}_t^i)^{-1 / 2} (\mathbf{X}^i_t -\widehat{\boldsymbol{\mu}}^i_t)$$ could be used as proxies for realized innovations.

\textbf{Step 2} (Dynamic CoVaR estimation): By the definition of CoVaR, the estimate of CoVaR at level $p$ and time $t$ for institution $i$ and system proxy $Y$ is given by
\begin{equation}
\widehat{\operatorname{CoVaR}}_{Y_t\mid X_t^i}(p)=\widehat{\mu}_t^Y+\widehat{\sigma}_{t, Y} \widehat{\operatorname{CoVaR}}_{\widehat{Z}_t^Y\mid \widehat{Z}_t^i}(p).
\label{dycov}\end{equation}

We would like to point out that \cite{nolde2022extreme} only applys the GARCH filters marginally because filtering out correlation led to weaker dependence in tail regions, invalidating the assumption of asymptotic dependence assumption in their article, while our new approach is able to deal with asymptotic independence. 
\subsection{Results}
In this subsection, we perform a dynamic analysis to assess accuracy of out-of-sample forecasts of CoVaR at risk level $p=5\%$. 
We use a rolling window of 3000 data points to estimate model parameters in Step 1 and investigate the asymptotic dependence structure for each pair of $\widehat{\mathbf{Z}}_t^i=\left(\widehat{Z}_t^i, \widehat{Z}_t^Y\right), i=1, \cdots, 14$, before producing one-day ahead CoVaR forecasts in Step 2. Here, we apply the Tail Quotient Correlation Coefficient (TQCC) test in \cite{zhang2017random} to test the null hypothesis of asymptotic independence, which corresponds to the case $\eta\in(0,1)$. 

To conduct the test, we first fit generalized extreme value distribution to each series and perform marginal transformations.  We apply TQCC to the transformed data, and choose the threshold as the smaller one of the two empirical 95th percentiles. For more details on TQCC-test, we refer to \cite{zhang2017random}. The computed TQCC measures and $p$-values are summarized in Table \ref{tb1}. We can find that while there are four institutions AFL, HUM, SLM and UNM that reject the null hypothesis of asymptotic independence with $p$-value below 0.01, the other eleven institutions (AIG, ALL, BAC, C, CMA, JPM, LNC, PGR, TRV, WFC, WM) demonstrate an asymptotic independence structure with the system proxy. We proceed with Step 2 for the eleven institutions. To reduce computational time, CoVaR estimates based on the samples of realized innovations are updated only every 50 observations. 
\begin{table}[]
  \centering
  \caption{TQCC, $p$-value, average quantile scores of CoVaR forecasts}\begin{tabular}{cccccc}
    \hline
    Ticker & TQCC & $p$-value & Model 1 & Model 2 \\ \hline
    AFL    & 0.72 & 0.00      &    &  \\
    AIG    & 0.00 & 1.00      & 1.01    & 0.74    \\
    ALL    & 0.00 & 1.00      & 1.10    & 0.90    \\
    BAC    & 0.00 & 0.76      & 1.19    & 0.94    \\
    C      & 0.00 & 1.00      & 1.17    & 0.97    \\
    CMA    & 0.00 & 1.00      & 1.02    & 0.89    \\
    HUM    & 0.63 & 0.00      &  &    \\
    JPM    & 0.00 & 1.00      & 1.20    & 0.95    \\
    LNC    & 0.00 & 1.00      & 1.23    & 1.01    \\
    PGR    & 0.00 & 1.00      & 0.97    & 0.77    \\
    SLM    & 0.80 & 0.00      &   &     \\
    TRV    & 0.00 & 0.99      & 1.04    & 0.95    \\
    UNM    & 0.88 & 0.00      &   &    \\
    WFC    & 0.00 & 1.00      & 1.04    & 0.87    \\
    WM     & 0.00 & 1.00      &  1.10  & 0.85   \\ \hline
\end{tabular}
\label{tb1}
\end{table}
  
Estimation of $\operatorname{CoVaR}_{\widehat{Z}_t^Y\mid \widehat{Z}_t^i}(p)$ requires choosing suitable values for sample fractions $k_1, k_2$ and $k_3$. We select $k_1=150$ for the tail index of $\widehat{Z}_t^Y$ by the Hill plot and $k_2=250$ from the stable region of the sensitivity analysis of $\widehat{\operatorname{VaR}}_{\widehat{Z}_t^Y}(p)$ to values of $k_2$. To choose $k_3$, we plot the relationship between $k_3$ values and the corresponding estimates for $\eta$ for all institutions and the index in Figure \ref{fig1}. It shows that the curves seems to be stable for a wide range of values. We select $k_3 =400$ for the eleven institutions.

\begin{figure}[htbp]
    \centering
    
    \begin{subfigure}[b]{0.3\textwidth}
        \includegraphics[width=\textwidth]{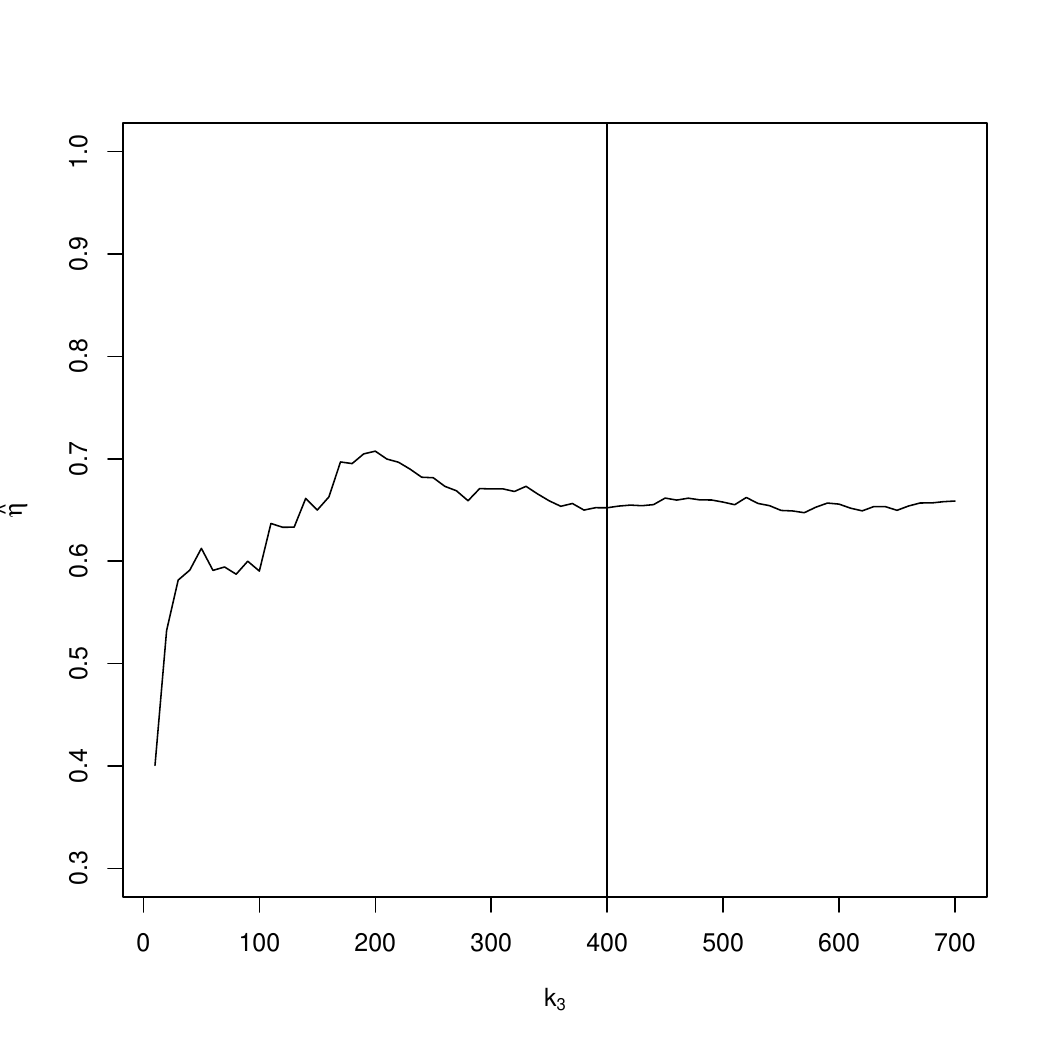}
        \caption{AIG}
    \end{subfigure}
    \hfill
    \begin{subfigure}[b]{0.3\textwidth}
        \includegraphics[width=\textwidth]{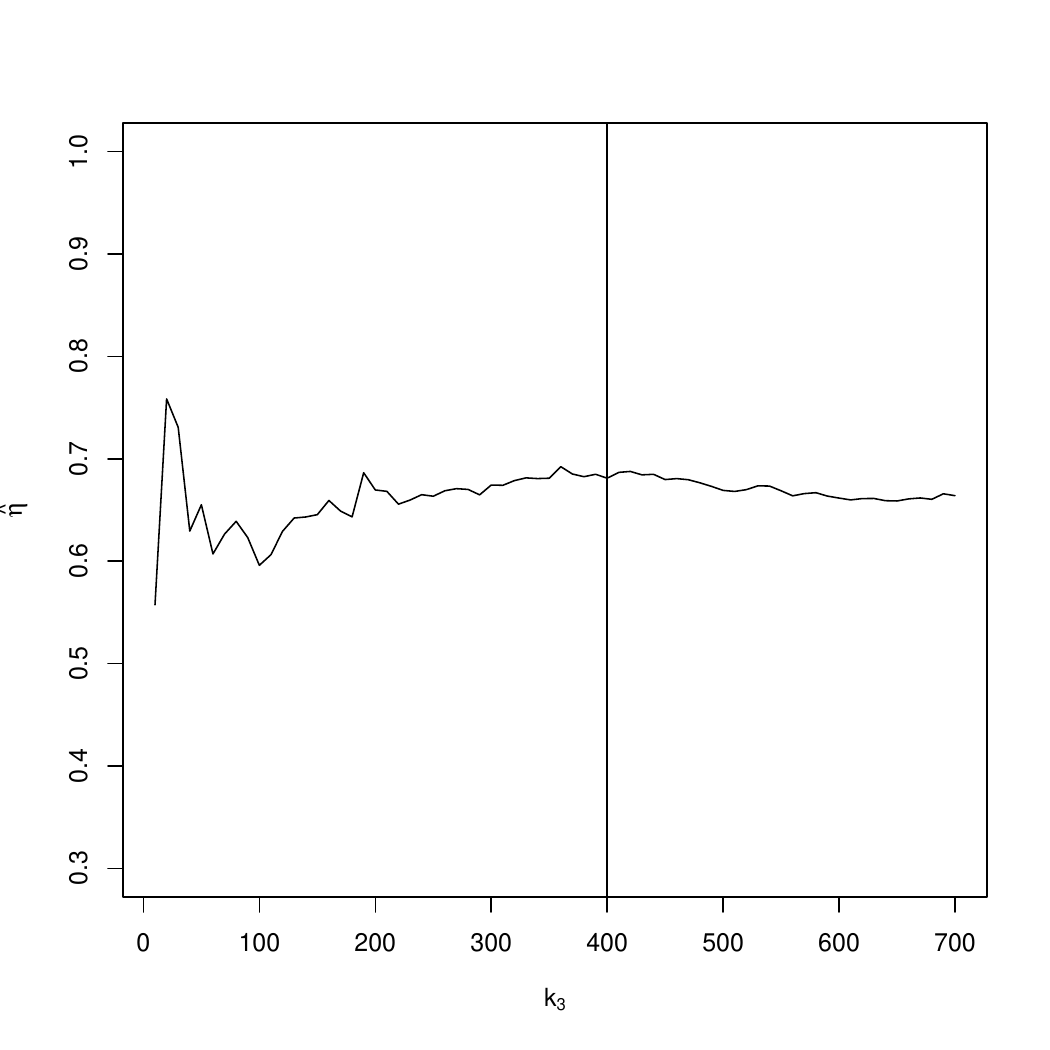}
        \caption{ALL}
      \end{subfigure}
        \hfill
        \begin{subfigure}[b]{0.3\textwidth}
          \includegraphics[width=\textwidth]{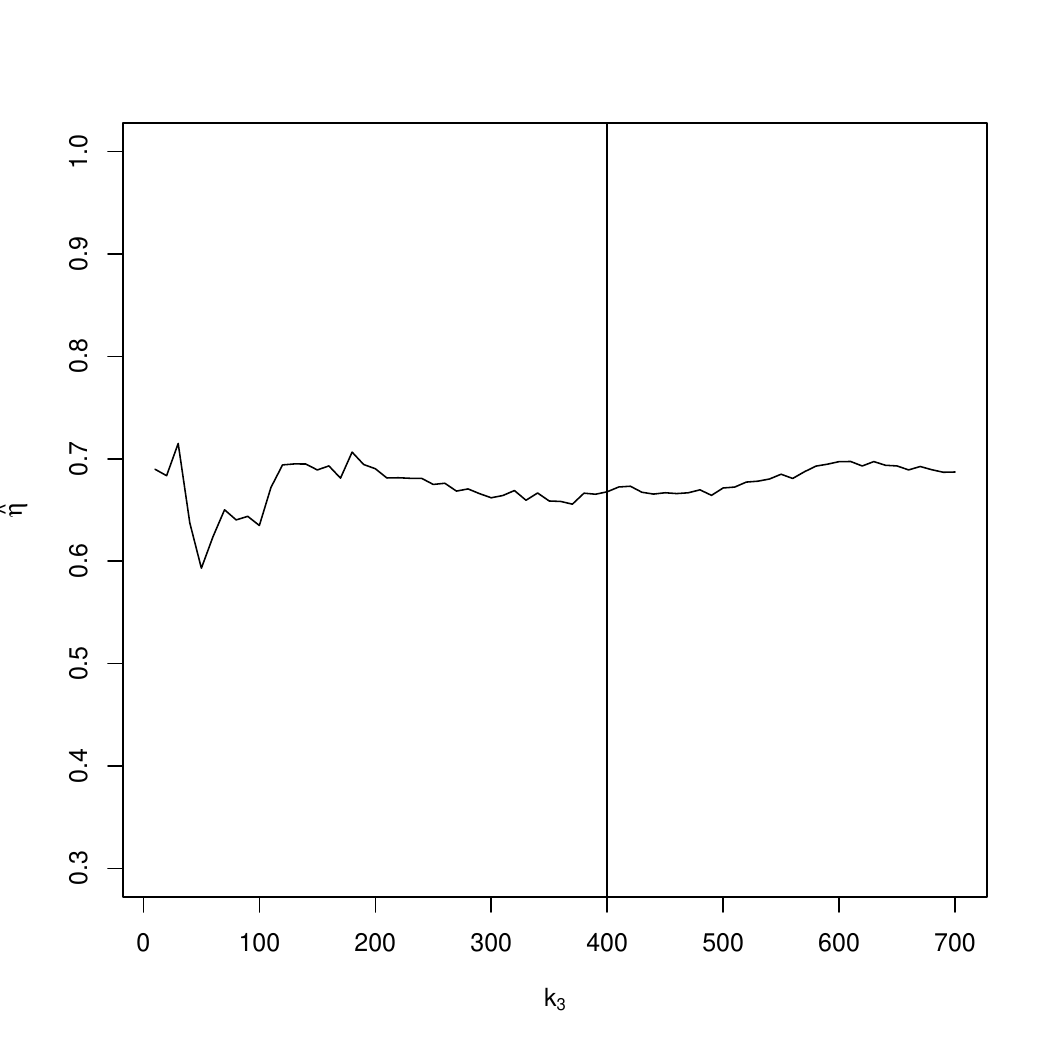}
          \caption{BAC}
        \end{subfigure}
          \hfill
          \begin{subfigure}[b]{0.3\textwidth}
            \includegraphics[width=\textwidth]{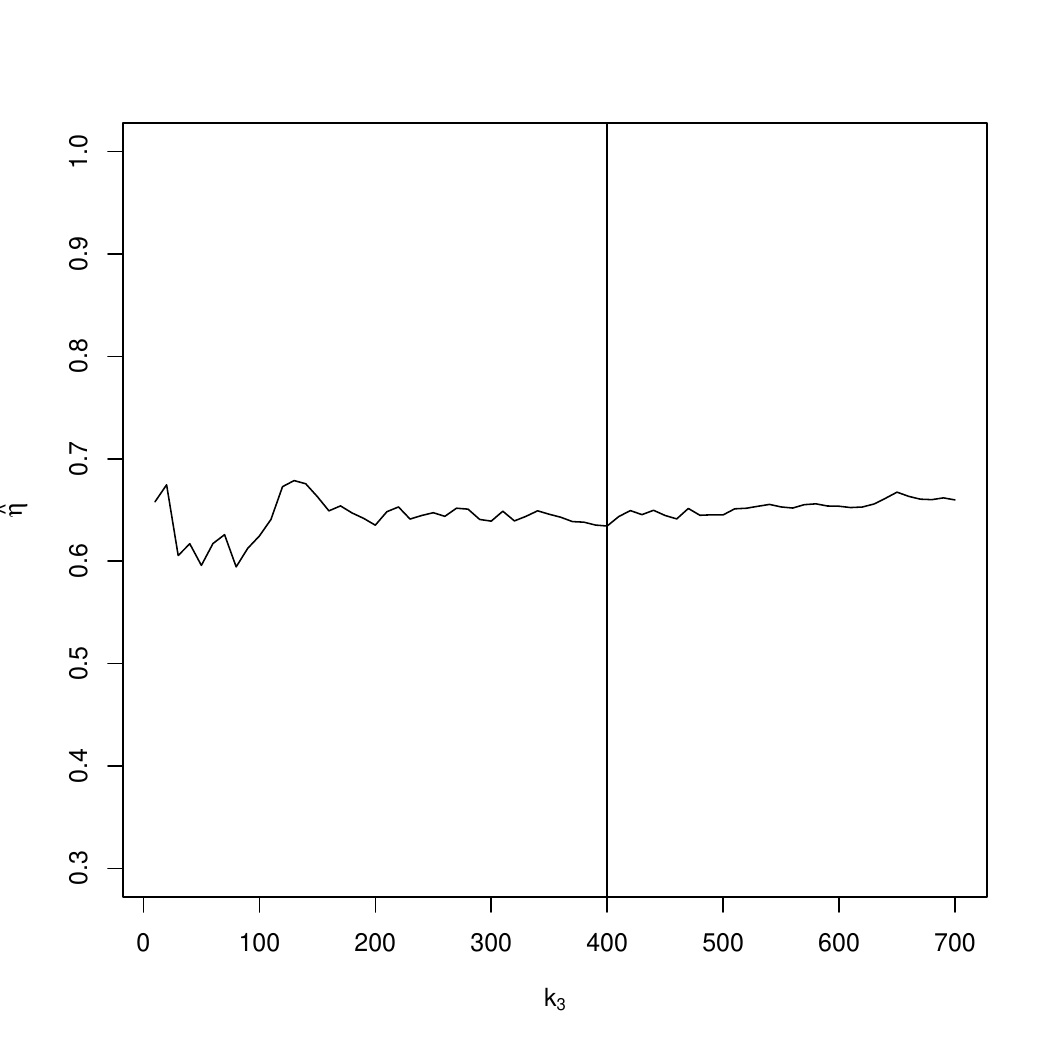}
            \caption{C}
          \end{subfigure}
            \hfill
            \begin{subfigure}[b]{0.3\textwidth}
              \includegraphics[width=\textwidth]{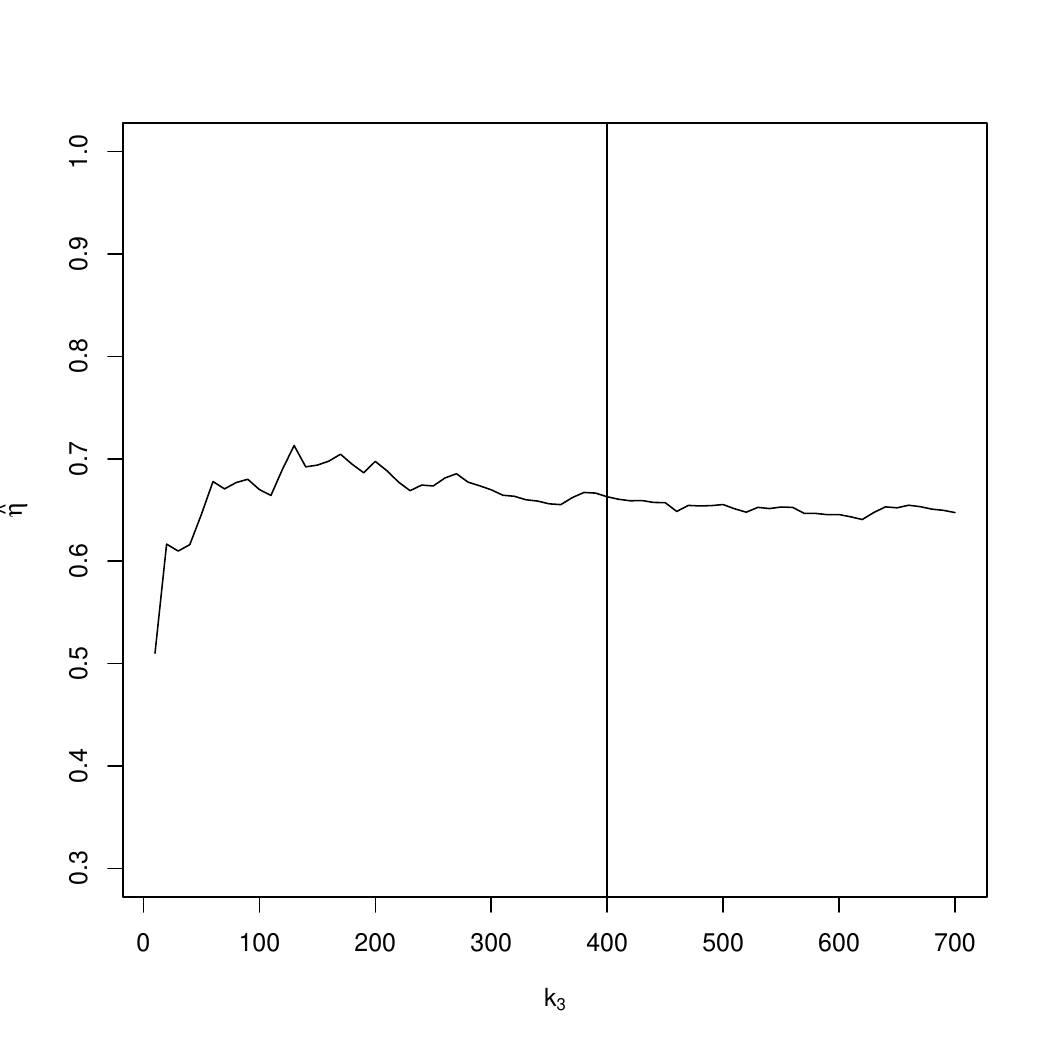}\caption{CMA}
            \end{subfigure}
              \hfill
              \begin{subfigure}[b]{0.3\textwidth}
                \includegraphics[width=\textwidth]{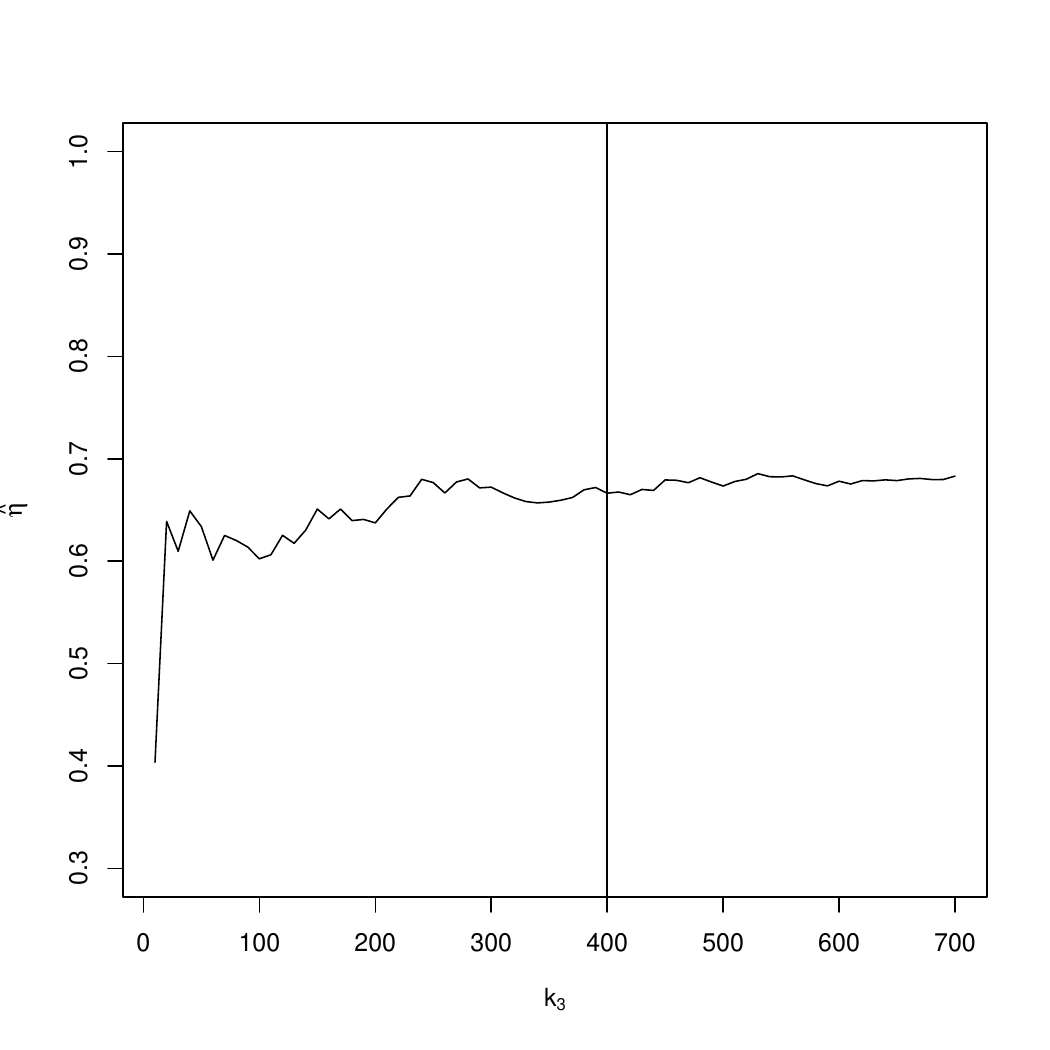}
              \caption{JPM}
            \end{subfigure}
              \hfill
              \begin{subfigure}[b]{0.3\textwidth}
                \includegraphics[width=\textwidth]{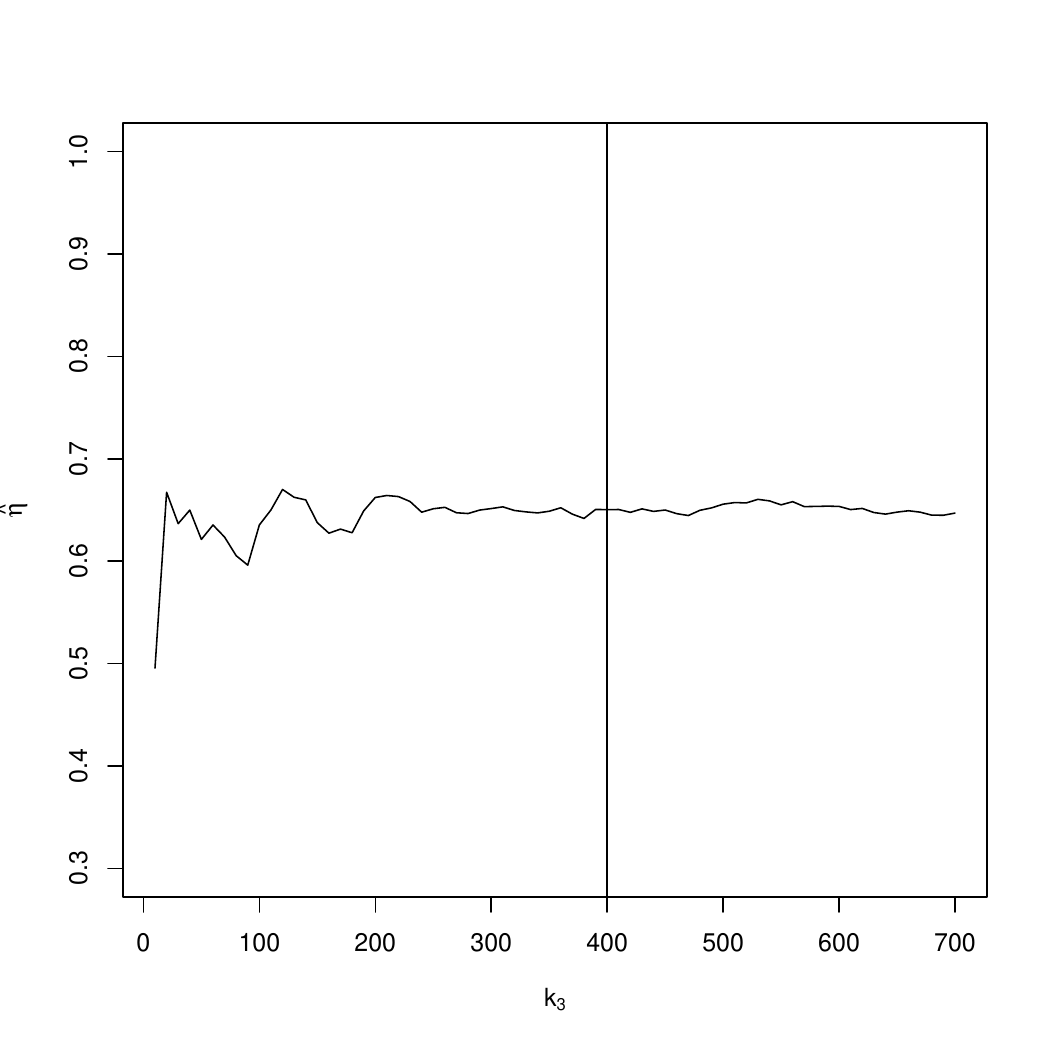}
                \caption{LNC}
              \end{subfigure}
                \hfill
                \begin{subfigure}[b]{0.3\textwidth}
                  \includegraphics[width=\textwidth]{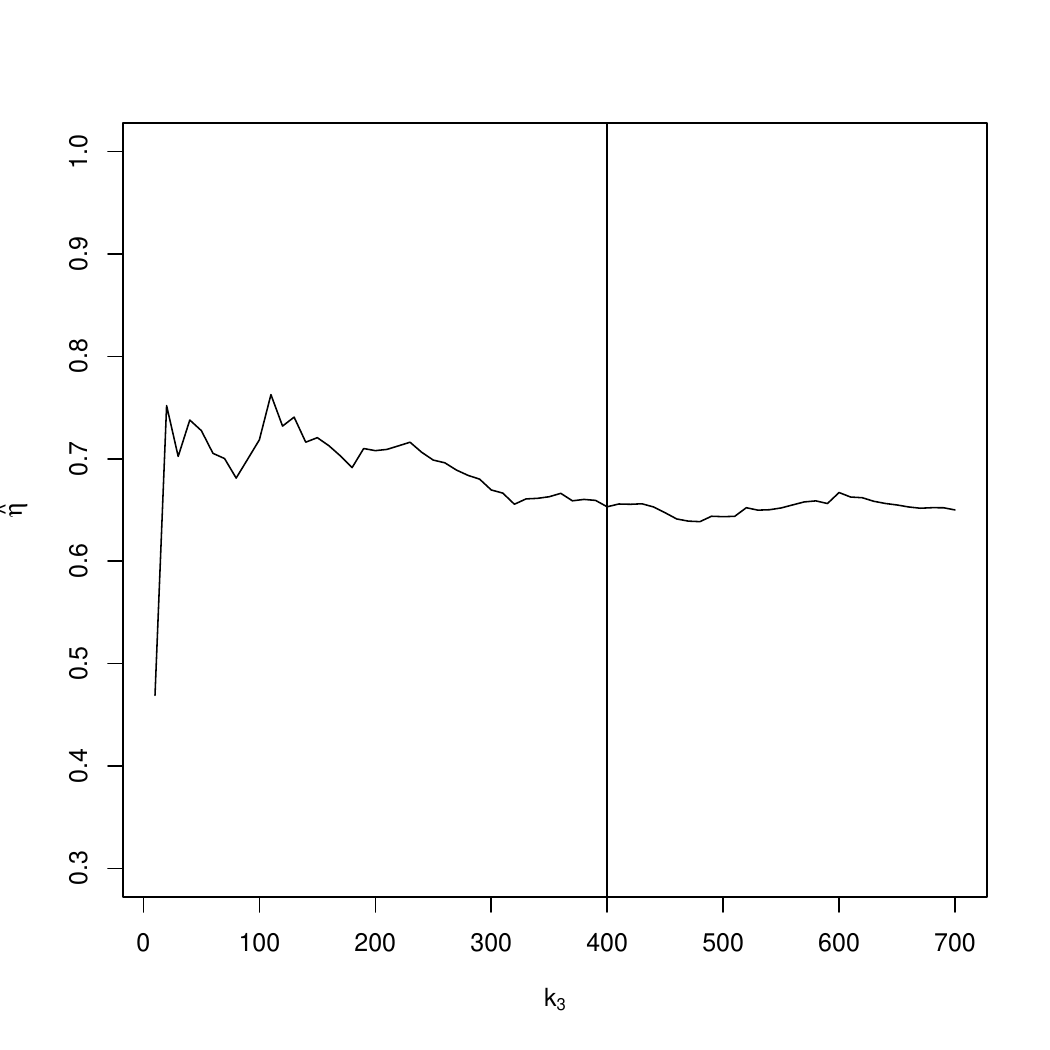}
                \caption{PGR}
              \end{subfigure}
                \hfill
                \begin{subfigure}[b]{0.3\textwidth}
                  \includegraphics[width=\textwidth]{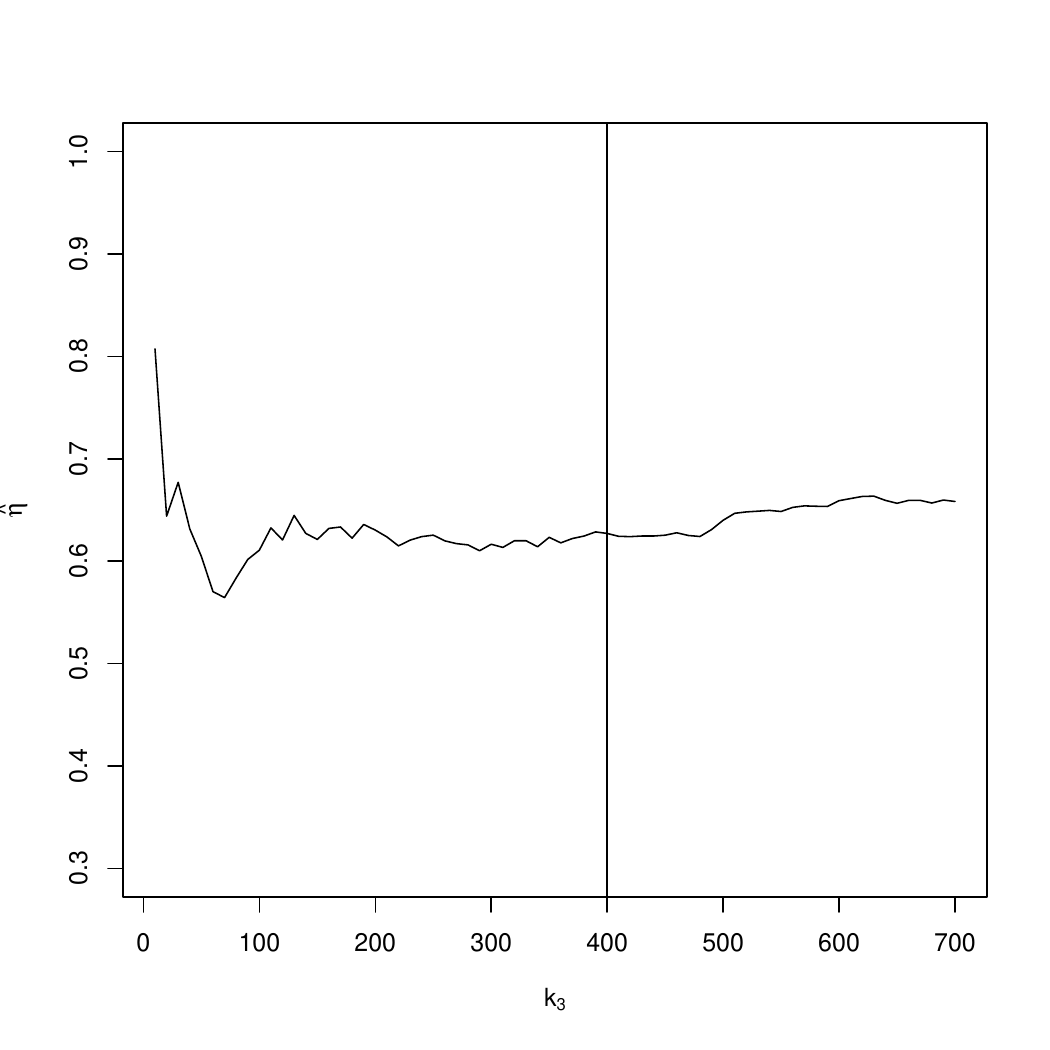}
                  \caption{TRV}
                \end{subfigure}
                  \hfill
                  \begin{subfigure}[b]{0.3\textwidth}
                    \includegraphics[width=\textwidth]{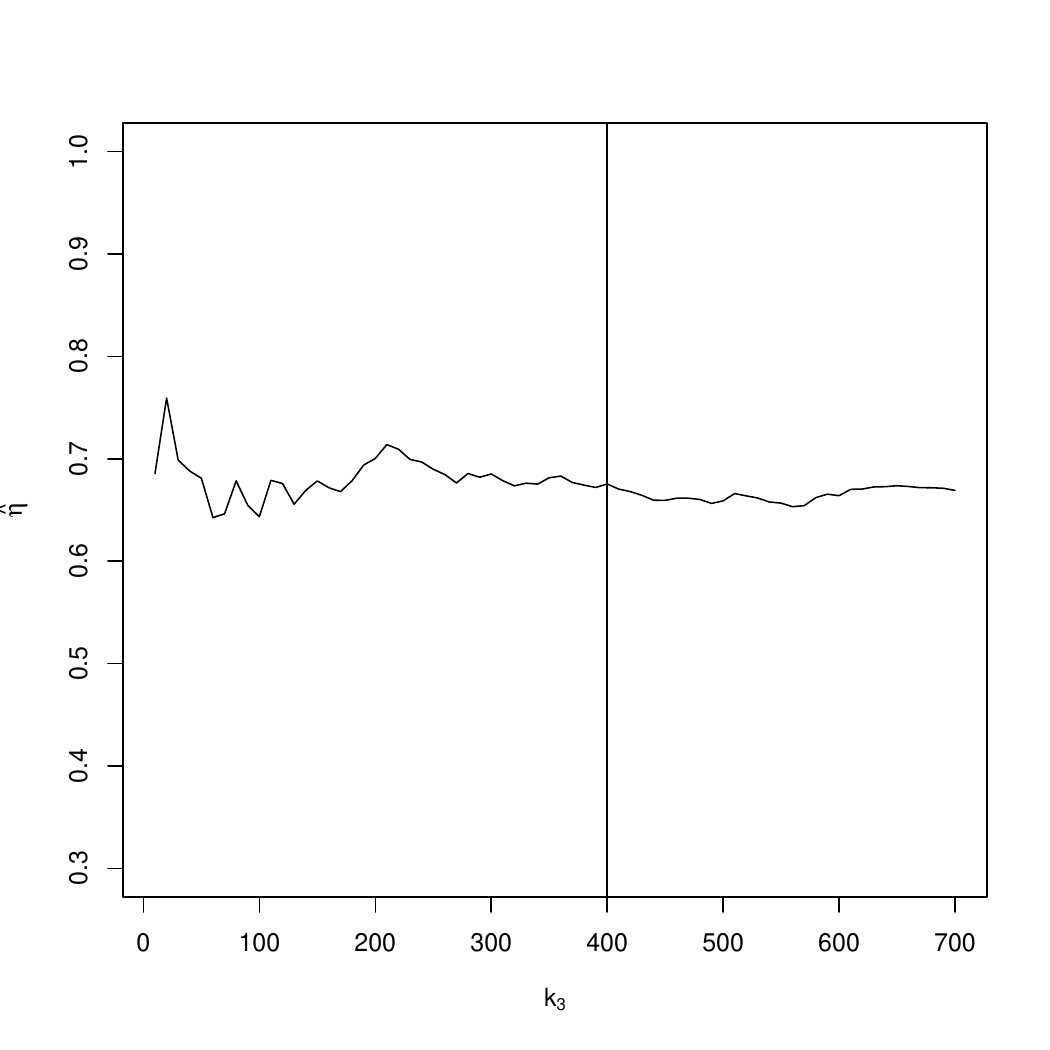}
                    \caption{WFC}
                  \end{subfigure}
                    \hfill
                    \begin{subfigure}[b]{0.3\textwidth}
                      \includegraphics[width=\textwidth]{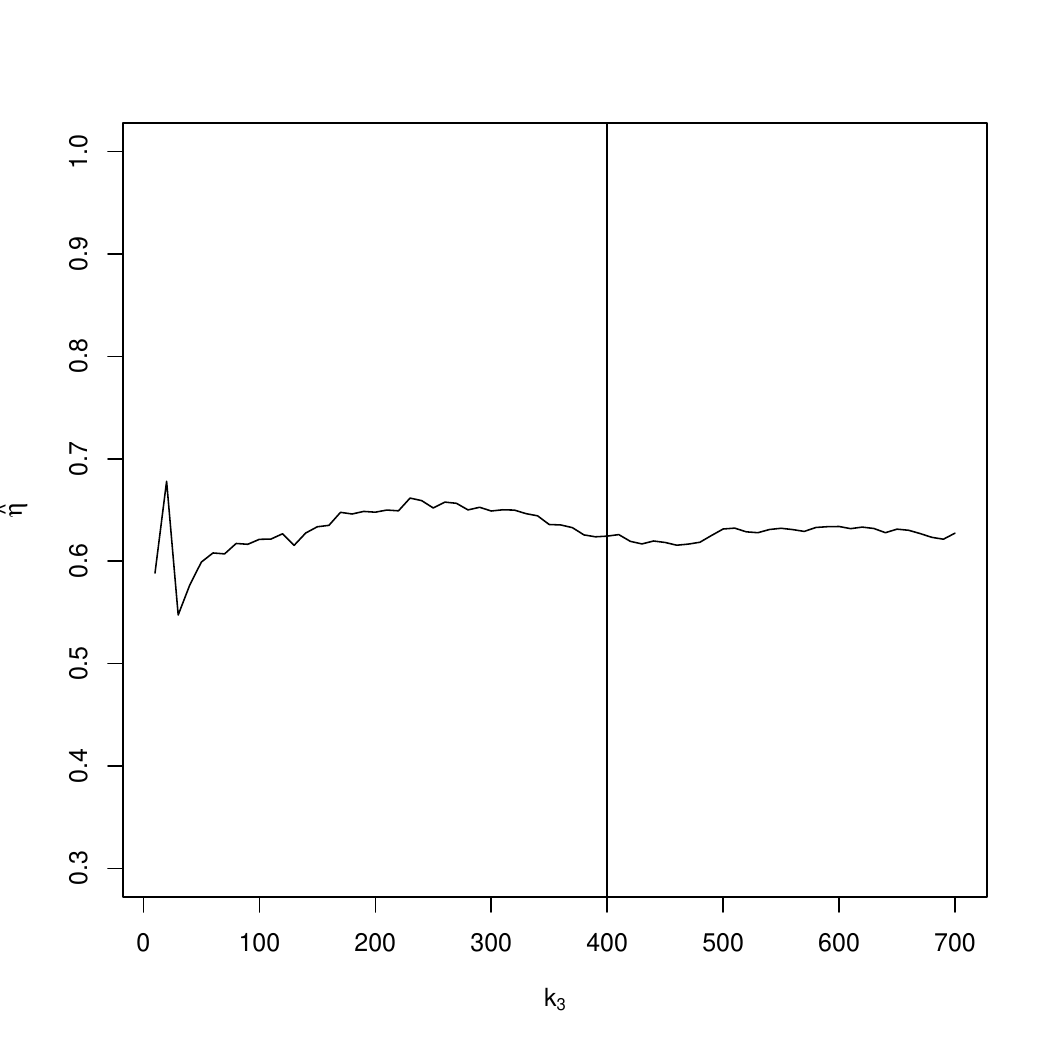}
                      \caption{WM}
                  \end{subfigure}

    \caption{Estimates of $\eta$ as a function of $k_3$ for the eleven institutions. }
    \label{fig1}
\end{figure}
By noting that CoVaR is itself VaR, or a quantile, for the conditional loss distribution of the system given that an institution $j$ exceeds its VaR, we make further comparisons across the considered method by calculating the average quantile score with the classical 1-homogeneous scoring function for $\left(1-p\right)$-quantile: $S(r, x)=\left(p_2-1\{x>r\}\right) r+1\{x>r\} x$, where $r$ denotes the estimate or forecast and $x$ the observation, and apply them to the sub-sample that coincide with the distress event that the
institution's losses being above its VaR. To keep consistency across different methods and parametric models, we take the empirical quantiles at level $1-p$ for each rolling window as the VaR estimates. The calculated average quantile scores are summarized in Table \ref{tb1}. When comparing within the proposed methodology, Model 2 leads to a superior forecasting performance relative to Model 1 for all of the financial institutions. 


\section{Proofs}\label{pf}Before proving Theorem \ref{thm1}, we present the following two lemmas, providing some preliminary results regarding  $\eta_{p}^*$ as well as the estimator $\widehat{\eta}_p^*$.
\begin{lemma}\label{lem1}Under the same conditions as in Theorem \ref{thm1}, we have that, as $n \rightarrow \infty$,
    $$
    \frac{p^{2-1/\eta}}{\eta^*_p} \rightarrow c_y\left(1,0 ; \boldsymbol{\theta}_0\right) \text { and }  \frac{p^{2-1/\eta}}{\widehat{\eta}^*_p} \stackrel{\mathbb{P}}{\rightarrow} c_y\left(1,0 ; \boldsymbol{\theta}_0\right).
    $$\end{lemma}\noindent\textbf{
    Proof.} In order to prove the limit relation regarding $\eta^*_p$, we first show that as $n \rightarrow \infty, \eta^*_p \rightarrow 0$. If otherwise, there exists a subsequence of integers $\left\{n_l\right\}$, such that $\eta^*_{p\left(n_l\right)} \rightarrow d>0$ as $l \rightarrow \infty$. Meanwhile, as $l \rightarrow \infty$, $c\left(1, \eta^*_{p\left(n_l\right)} ; \boldsymbol{\theta}_0\right) \rightarrow c\left(1,d ; \boldsymbol{\theta}_0\right)>0$ which follows from Assumption \ref{asm6} and the fact that $c\left(1, y ; \boldsymbol{\theta}_0\right)$ is a non-decreasing function in $y$. However, this contradicts with $c\left(1, \eta^*_{p\left(n_l\right)} ; \boldsymbol{\theta}_0\right)=(p\left(n_l\right))^{2-1/\eta} \rightarrow 0$ as $n \rightarrow \infty$. Hence, we conclude that $\eta_{p}^* \rightarrow 0$ as $n \rightarrow \infty$.

Using the mean value theorem, we have that there exists a series of constants $\xi_n \in\left[0, \eta_{p}^*\right]$ such that
$$
p^{2-1/\eta}=c\left(1, \eta_{p}^* ; \boldsymbol{\theta}_0\right)=c\left(1,0 ; \boldsymbol{\theta}_0\right)+\eta_{p}^* c_y\left(1, \xi_n ; \boldsymbol{\theta}_0\right)=\eta_{p}^* c_y\left(1, \xi_n ; \boldsymbol{\theta}_0\right).
$$
Hence we get that, as $n \rightarrow \infty$,
$$
\frac{p^{2-1/\eta}}{\eta_{p}^*}=c_y\left(1, \xi_n ; \boldsymbol{\theta}_0\right) \rightarrow c_y\left(1,0 ; \boldsymbol{\theta}_0\right),
$$
by the fact that $\xi_n \rightarrow 0$ as $n \rightarrow \infty$ and that $c_y\left(x, y ; \boldsymbol{\theta}_0\right)$ is a continuous function at $\left(1,0 ; \boldsymbol{\theta}_0\right)$.

For the limit relation regarding $\widehat{\eta}_p^*$, we need to  replace $\boldsymbol{\theta}_0$ and $\eta$ by $\widehat{\boldsymbol{\theta}}$ and $\widehat{\eta}$, respectively. Based on Theorem 3 in \cite{lalancette2021rank}, we have \begin{equation}\sqrt{m}(\widehat{\boldsymbol{\theta}}-\boldsymbol{\theta}_0)\stackrel{d}{\rightarrow}  N(0,\sigma_{\theta}^2),\label{theta}\end{equation}for some $\sigma^2_{\theta}>0,$ where $m=n(k_3/n)^{1/\eta}$, defined in Assumption \ref{asm4+}. By applying Delta method to $\widehat{\eta}=\eta(\widehat{\boldsymbol{\theta}})$, we have $$\sqrt{m}(\widehat{\eta}-\eta)\stackrel{d}{\rightarrow} N(0,\sigma_{\eta}^2),$$for some $\sigma^2_{\eta}>0.$ 
The rest of the proof follows similarly and we therefore omit the details.

$\hfill\qedsymbol$

\begin{lemma} \label{lem2}Under the same conditions as in Theorem \ref{thm1}, we have, as $n \rightarrow \infty$,
    $$
    \frac{\eta_{p}}{\eta_{p}^*} \rightarrow 1.
    $$
\end{lemma}

\noindent\textbf{Proof.} We first show that, as $n \rightarrow \infty, \eta_{p} \rightarrow 0$. Otherwise, there exists a subsequence of integers $\left\{n_l\right\}$, such that $\eta_{p\left(n_l\right)} \rightarrow d>0$ as $l \rightarrow \infty$. Recall that by the definition of $\eta_{p}$ we have $$
Q\left(p, p\eta_{p}\right)/p^{1-1/\eta}=p^{2-1/\eta}.$$By taking $l \rightarrow \infty$ on both sides of this equation, and using the assumption that $\eta_{p\left(n_l\right)} \rightarrow d>0$ as $n \rightarrow \infty$, we get that $c\left(1, d; \boldsymbol{\theta}_0\right)=0$, which contradicts Assumption \ref{asm6} and the fact that $c\left(1, y ; \boldsymbol{\theta}_0\right)$ is a non-decreasing function in $y$. Hence, we conclude that, as $n \rightarrow \infty, \eta_{p} \rightarrow 0$.

Next we show, by contradiction, that
$$
\limsup _{n \rightarrow \infty} \frac{\eta_{p}}{\eta_{p}^*} \leq 1.
$$If assuming otherwise, there exists a subsequence of $\left\{n_l\right\}_{l=1}^{\infty}$ such that as $l \rightarrow \infty, n_l \rightarrow \infty$ and
$$
\frac{\eta_{p\left(n_l\right)}}{\eta_{p\left(n_l\right)}^*} \rightarrow d>1.
$$
 Without loss of generality (W.l.o.g.), we use the notation $n$ instead of $n_l$ and write $p=p(n)$. Therefore, for any $1<\tilde{d}<d$, there exists $n_0=n_0(\tilde{d})$ such that for $n>n_0, \eta_{p}/\eta_{p}^*>\tilde{d}$.

Note that $\eta_{p}^*<\tilde{d} \eta_{p}^*<\eta_{p}$. By the mean value theorem, we get that for each $n$, there exists $\xi_n \in\left(\eta_{p}^*, \eta_{p}\right)$ such that
\begin{equation}
  \label{der}
c\left(1, \eta_{p} ; \boldsymbol{\theta}_0\right)-c\left(1, \eta_{p}^* ; \boldsymbol{\theta}_0\right)=c_y\left(1, \xi_n ; \boldsymbol{\theta}_0\right)\left(\eta_{p}-\eta_{p}^*\right).
\end{equation}

As $n \rightarrow \infty$, since both $\eta_{p}^* \rightarrow 0$ and $\eta_{p} \rightarrow 0$ hold, we get $\xi_n \rightarrow 0$. Further note that $\eta_{p}-\eta_{p}^*>(\tilde{d}-1) \eta_{p}^*$. By applying Lemma \ref{lem1} and the continuity of $c_y(x, y ; \boldsymbol{\theta})$ at $\left(1,0 ; \boldsymbol{\theta}_0\right)$, we obtain that

\begin{align}
\liminf _{n \rightarrow \infty} p^{1/\eta-2}\left(c\left(1, \eta_{p} ; \boldsymbol{\theta}_0\right)- p^{2-1/\eta}\right) & =\liminf _{n \rightarrow \infty}p^{1/\eta-2}\left(c\left(1, \eta_{p} ; \boldsymbol{\theta}_0\right)-c\left(1, \eta_{p}^* ; \boldsymbol{\theta}_0\right)\right)\nonumber \\
& =\liminf _{n \rightarrow \infty}\eta_{p}^*p^{1/\eta-2}\frac{c\left(1, \eta_{p} ; \boldsymbol{\theta}_0\right)-c\left(1, \eta_{p}^* ; \boldsymbol{\theta}_0\right)}{\eta_{p}^*}\nonumber \\
& \geq \frac{1}{c_y\left(1,0 ; \boldsymbol{\theta}_0\right)}c_y\left(1,0 ; \boldsymbol{\theta}_0\right)(\tilde{d}-1) =\tilde{d}-1>0.\label{liminf}
\end{align}On the other side, by Assumption \ref{asm4}, it follows that
\begin{align}p^{1/\eta-2}\left(c\left(1, \eta_{p} ; \boldsymbol{\theta}_0\right)- p^{2-1/\eta}\right)=&p^{1/\eta-2}\left(c\left(1, \eta_{p} ; \boldsymbol{\theta}_0\right)- \frac{Q(p,p\eta_{p})}{p^{1/\eta}}\right)\nonumber\\=&p^{1/\eta-2}O(q_1(p))=o(1).\label{limsup}
\end{align}
The two limit relations (\ref{liminf}) and (\ref{limsup}) contradict each other. Therefore, we conclude that
$$
\limsup _{n \rightarrow \infty} \frac{\eta_{p}}{\eta_{p}^*} \leq 1.
$$

Similarly, one can show a lower bound for $\eta_{p}/\eta_{p}^*$, which completes the proof of the lemma.$\hfill\qedsymbol$

\bigskip\noindent\textbf{Proof of Theorem \ref{thm1}.} Based on the definition of $\eta_{p}$, we rewrite 
    \begin{align}
        \frac{\widehat{\operatorname{CoVaR}}_{Y \mid X}(p)}{\operatorname{CoVaR}_{Y \mid X}(p)} & =\frac{\left(\widehat{\eta}_p^*\right)^{-\widehat{\gamma}} \widehat{\operatorname{VaR}}_Y(p)}{\operatorname{VaR}_{Y}\left(p \eta_{p}\right)} \nonumber\\
    & =\left(\frac{\widehat{\eta}_p^*}{\eta_{p}}\right)^{-\widehat{\gamma}} \times \eta_{p}^{\gamma-\widehat{\gamma}} \times \frac{\widehat{\operatorname{VaR}}_Y(p)}{\operatorname{VaR}_Y(p)} \times \frac{\left(\eta_{p}\right)^{-\gamma} \operatorname{VaR}_Y(p)}{\operatorname{VaR}_Y\left(p \eta_{p}\right)} \nonumber\\
    & =: I_1 \times I_2 \times I_3 \times I_4 .\label{dec}
    \end{align}
    We will show that $I_j \stackrel{\mathbb{P}}{\rightarrow} 1,$ for $j=1,2,3,4$, as $n \rightarrow \infty$.

    Firstly, we handle $I_1$. Following the asymptotic property of the Hill estimator (e.g., Theorem 3.2.5 in \cite{de2007extreme}), Assumptions \ref{asm1} and \ref{asm2} for $k_1$ imply that as $n \rightarrow \infty$,
\begin{equation}
    \label{gam}\sqrt{k_1}(\widehat{\gamma}-\gamma) \stackrel{d}{\rightarrow} N\left(\frac{\lambda_1}{1-\rho}, \gamma^2\right),
\end{equation}
which implies that $\widehat{\gamma} \stackrel{\mathbb{P}}{\rightarrow} \gamma$. Together with Lemmas \ref{lem1} and \ref{lem2}, we conclude that $I_1 \stackrel{\mathbb{P}}{\rightarrow} 1$ as $n \rightarrow \infty$.

Secondly, we handle $I_2$. Given the limit relation (\ref{gam}), we only need to show that $\log \left(\eta_{p}\right) / \sqrt{k_1} \rightarrow 0$ as $n \rightarrow \infty$. From Lemmas \ref{lem1} and \ref{lem2}, we obtain $p^{1-1/\eta}\eta_{p}  \rightarrow1 / c_y\left(1,0; \boldsymbol{\theta}_0\right)$ as $n \rightarrow \infty$. Together with the limit relation regarding $k_1$ in Assumption \ref{asm2}, we get that $I_2 \stackrel{\mathbb{P}}{\rightarrow} 1$ as $n \rightarrow \infty$.

The term $I_3$ is handled by the asymptotic property of the VaR estimator; see, e.g. Theorem 4.3.8 in \cite{de2007extreme}. More specifically, under Assumptions \ref{asm1} to \ref{asm3}, the VaR estimator in (\ref{var}) has the following asymptotic property:
    $$
    \min \left(\sqrt{k_2}, \frac{\sqrt{k_1}}{\log \left(k_2 / n p\right)}\right)\left(\frac{\widehat{\operatorname{VaR}}_Y(p)}{\operatorname{VaR}_Y(p)}-1\right)=O_{\mathbb{P}}(1) .
    $$A direct consequence is that $I_3 \stackrel{\mathbb{P}}{\rightarrow} 1$ as $n \rightarrow \infty$.
    
    Finally, we handle the deterministic term $I_4$. Notice that $\operatorname{VaR}_Y(p)=U_2(1 / p)$ and $\operatorname{VaR}_Y\left(p \eta_{p}\right)=U_2\left(1 /\left(p \eta_{p}\right)\right)$. By applying Assumption \ref{asm1} with $t=1 / p$ and $x=1 / \eta_{p}$, we get that
    $$
    \lim _{n \rightarrow \infty} \frac{\frac{\operatorname{VaR}_Y\left(p \eta_{p}\right)}{\operatorname{VaR}_Y(p)} \eta_{p}^\gamma-1}{A(1 / p)}=-\frac{1}{\rho} .
    $$
    As $n \rightarrow \infty$, since $A(1 / p) \rightarrow 0$, we get that $I_4 \rightarrow 1$. The proof is complete.$\hfill\qedsymbol$  

    \bigskip
    Before proving Theorem \ref{thm2}, we present the following two lemmas, providing further asymptotic results regarding  $\eta_{p}^*$ as well as the estimator $\widehat{\eta}_p^*$. Note that Lemma \ref{lem2+} is a strengthened version of Lemma \ref{lem2}.
    \begin{lemma} \label{lem2+}Under the same conditions as in Theorem \ref{thm2}, we have, as $n \rightarrow \infty$,
      $$
\frac{\sqrt{k_1}}{\log (k_2/np)}\left(\frac{\eta_{p}}{\eta_{p}^*}-1\right)\rightarrow 0.
$$
  \end{lemma}
  
  \noindent\textbf{Proof.} Recall that by(\ref{der}): $$c_y\left(1, \xi_n ; \boldsymbol{\theta}_0\right)\left(\eta_{p}^*-\eta_{p}\right)=c\left(1, \eta_{p}^* ; \boldsymbol{\theta}_0\right)-c\left(1, \eta_{p} ; \boldsymbol{\theta}_0\right)=p^{2-1/\eta}-c\left(1, \eta_{p} ; \boldsymbol{\theta}_0\right).$$Therefore, we get$$\frac{\sqrt{k_1}}{\log (k_2/np)}\left(\frac{\eta_{p}}{\eta_{p}^*}-1\right)=\frac{1}{c_y\left(1, \xi_n ; \boldsymbol{\theta}_0\right)}\frac{p^{2-1/\eta}}{\eta_{p}}\frac{\sqrt{k_1}}{\log (k_2/np)}\frac{p^{2-1/\eta}-c\left(1, \eta_{p} ; \boldsymbol{\theta}_0\right)}{p^{2-1/\eta}}.$$
  As $n\to\infty$, $\frac{1}{c_y\left(1, \xi_n ; \boldsymbol{\theta}_0\right)}\to\frac{1}{c_y\left(1, 0 ; \boldsymbol{\theta}_0\right)}$ and $\frac{p^{2-1/\eta}}{\eta_{p}}=O(1)$. Recalling the definition of $\eta_{p}$ and Assumption \ref{asm8}, we get that$$\frac{\sqrt{k_1}}{\log (k_2/np)}\frac{p^{2-1/\eta}-c\left(1, \eta_{p} ; \boldsymbol{\theta}_0\right)}{p^{2-1/\eta}}=\frac{\sqrt{k_1}}{\log (k_2/np)}\frac{q_1(p)}{p^{2-1/\eta}}\to0,$$as $n\to\infty$. The proof is complete.
  $\hfill\qedsymbol$  
    \begin{lemma} \label{lem3}Under the same conditions as in Theorem \ref{thm2}, we have, as $n \rightarrow \infty$,
      $$
      \frac{\sqrt{k_1}}{\log (k_2/np)}\left(\frac{\eta_{p}}{\widehat{\eta}_p^*}-1\right)\stackrel{\mathbb{P}}{\rightarrow} 0.
      $$
  \end{lemma}
  
  \noindent\textbf{Proof.}
  By Assumption \ref{asm7} and Lemma \ref{lem1}, we have that 
  $$
\frac{p^{2-1/\eta}}{\eta_{p}^*}=\frac{c\left(1, \eta_{p}^* ; \boldsymbol{\theta}_0\right) }{\eta_{p}^*}=c_y\left(1,0 ; \boldsymbol{\theta}_0\right) +O((\eta_{p}^*)^{\breve{\rho}(\boldsymbol{\theta}_0)})=c_y\left(1,0 ; \boldsymbol{\theta}_0\right) +O(p^{\breve{\rho}(\boldsymbol{\theta}_0)(2-1/\eta)}).
$$Under Assumption \ref{asm8}, we have $$\frac{\sqrt{k_1}}{\log (k_2/np)}\left(\frac{p^{2-1/\eta}}{\eta_{p}^*}-c_y\left(1,0 ; \boldsymbol{\theta}_0\right)\right)\to0.$$Similarly, given the fact that $\hat{\boldsymbol{\theta}}$ and $\hat{\eta}$ are both consistent estimators, $\breve{\rho}(\cdot)$ is a continuous function at $\boldsymbol{\theta}_0$ and $\sqrt{m} p^{\breve{\rho}(\boldsymbol{\theta}_0)(2-1/\eta)-\varepsilon}\rightarrow 0$ as $n \rightarrow \infty$, we have$$
\frac{\sqrt{k_1}}{\log (k_2/np)}\left(\frac{p^{2-1/\eta}}{\widehat{\eta}_p^*}-c_y\left(1,0 ;\widehat{ \boldsymbol{\theta}}\right) \right) \stackrel{\mathbb{P}}{\rightarrow} 0.$$Taking together, we get$$
\frac{\sqrt{k_1}}{\log (k_2/np)}\left(\frac{\eta_{p}^*}{\widehat{\eta}_p^*}-\frac{c_y\left(1,0 ;\widehat{ \boldsymbol{\theta}}\right)}{c_y\left(1,0 ; \boldsymbol{\theta}_0\right)} \right) \stackrel{\mathbb{P}}{\rightarrow} 0.$$ By 1-Liptschiz continuity and (\ref{theta}), we obtain $$|c_y\left(1,0 ;\widehat{ \boldsymbol{\theta}}\right)-c_y\left(1,0 ; \boldsymbol{\theta}_0\right)|\leq |\widehat{ \boldsymbol{\theta}}-\boldsymbol{\theta}_0|=O_{\mathbb{P}}(\frac{1}{\sqrt{m}}).$$ Since $\frac{\sqrt{k_1}}{\sqrt{m}\log \left(k_2 / n p\right)}\to 0$ and $c_y\left(1,0 ; \boldsymbol{\theta}_0\right)>0,$ it follows that $$
\frac{\sqrt{k_1}}{\log (k_2/np)}\left(\frac{c_y\left(1,0 ;\widehat{ \boldsymbol{\theta}}\right)}{c_y\left(1,0 ; \boldsymbol{\theta}_0\right)}-1 \right) \stackrel{\mathbb{P}}{\rightarrow} 0,$$which further implies that $$
\frac{\sqrt{k_1}}{\log (k_2/np)}\left(\frac{\eta_{p}^*}{\widehat{\eta}_p^*}-1\right)\stackrel{\mathbb{P}}{\rightarrow} 0.
$$Finally we apply Lemma \ref{lem2+} and conclude that as $n \rightarrow \infty$, $$\frac{\sqrt{k_1}}{\log (k_2/np)}\left(\frac{\eta_{p}}{\widehat{\eta}_p^*}-1\right)\stackrel{\mathbb{P}}{\rightarrow} 0.
      $$ $\hfill\qedsymbol$

\bigskip\noindent\textbf{Proof of Theorem \ref{thm2}.} Recall the decomposition (\ref{dec}) and note that $I_j \stackrel{\mathbb{P}}{\rightarrow}1,$ for $j=1,2,3,4$, as $n \rightarrow \infty$.  

For $I_1$, by Taylor expansion, Lemma \ref{lem3} and Assumption \ref{asm8}, we have \begin{align*}I_1&=\left(\frac{\widehat{\eta}_p^*}{\eta_{p}}\right)^{-\widehat{\gamma}}=1-\widehat{\gamma}\frac{\widehat{\eta}_p^*}{\eta_{p}} =1-(\gamma+O_{\mathbb{P}}(\frac{1}{\sqrt{k_1}}))o_{\mathbb{P}}\left(\frac{\log (k_2/np)}{\sqrt{k_1}}\right)\\&=1+o_{\mathbb{P}}\left(\frac{\log (k_2/np)}{\sqrt{k_1}}\right).\end{align*}

\noindent Similarly, we get
\begin{align*}I_2&=e^{(\gamma-\widehat{\gamma})\log(\eta_{p})}=1+(\gamma-\widehat{\gamma})\log(\eta_{p})+O_{\mathbb{P}}((\gamma-\widehat{\gamma})\log(\eta_{p}))\\&=1+O_{\mathbb{P}}(\frac{\log p}{\sqrt{k_1}})=1+o_{\mathbb{P}}(\frac{\log (k_2/np)}{\sqrt{k_1}}).\end{align*}
\indent For $I_3$, by Theorem 4.3.8 in \cite{de2007extreme}, we have
$$\min \left(\sqrt{k_2}, \frac{\sqrt{k_1}}{\log \left(k_2 / n p\right)}\right)\left(\frac{\widehat{\operatorname{VaR}}_Y(p)}{\operatorname{VaR}_Y(p)}-1\right)\stackrel{d}{\rightarrow} \mathcal{N}(\frac{\lambda}{1-\rho},\gamma^2).$$Together with Assumption \ref{asm8}, we have $$\frac{\sqrt{k_1}}{\log \left(k_2 / n p\right)}\left(I_3-1\right)\stackrel{d}{\rightarrow} \mathcal{N}(\frac{\lambda}{1-\rho},\gamma^2).$$

For $I_4$, by Assumption \ref{asm1}, \ref{asm2} and \ref{asm8}, we have
\begin{align*}I_4&=\frac{\left(\eta_{p}\right)^{-\gamma} \operatorname{VaR}_Y(p)}{\operatorname{VaR}_Y\left(p \eta_{p}\right)}=1+O(A(1/p))=1+o(A(n/k_2))\\&=1+o(1/\sqrt{k_2})=1+o(\frac{\log (k_2/np)}{\sqrt{k_1}}).\end{align*}

Finally we get $$ \frac{\sqrt{k_1}}{\log \left(k_2 / n p\right)}\left(\frac{\widehat{\operatorname{CoVaR}}_{Y \mid X}(p)}{\operatorname{CoVaR}_{Y \mid X}(p)}-1\right)\stackrel{d}{\rightarrow}\mathcal{N}(\frac{\lambda}{1-\rho},\gamma^2),$$ and the proof is complete.

$\hfill\qedsymbol$

\bibliographystyle{apalike}

\bibliography{ref}

@article{engle2018systemic,
  title={Systemic risk 10 years later},
  author={Engle, Robert},
  journal={Annual Review of Financial Economics},
  volume={10},
  number={1},
  pages={125--152},
  year={2018},
  publisher={Annual Reviews}
}

@article{he2019statistical,
  title={Statistical inference for a relative risk measure},
  author={He, Yi and Hou, Yanxi and Peng, Liang and Sheng, Jiliang},
  journal={Journal of Business \& Economic Statistics},
  volume={37},
  number={2},
  pages={301--311},
  year={2019},
  publisher={Taylor \& Francis}
}

@article{brownlees2017srisk,
  title={SRISK: A conditional capital shortfall measure of systemic risk},
  author={Brownlees, Christian and Engle, Robert F},
  journal={The Review of Financial Studies},
  volume={30},
  number={1},
  pages={48--79},
  year={2017},
  publisher={Oxford University Press}
}

@article{allen2012does,
  title={Does systemic risk in the financial sector predict future economic downturns?},
  author={Allen, Linda and Bali, Turan G and Tang, Yi},
  journal={The Review of Financial Studies},
  volume={25},
  number={10},
  pages={3000--3036},
  year={2012},
  publisher={Oxford University Press}
}

@incollection{hansen2013challenges,
  title={Challenges in identifying and measuring systemic risk},
  author={Hansen, Lars Peter},
  booktitle={Risk topography: Systemic risk and macro modeling},
  pages={15--30},
  year={2013},
  publisher={University of Chicago Press}
}

@article{jackson2021systemic,
  title={Systemic risk in financial networks: A survey},
  author={Jackson, Matthew O and Pernoud, Agathe},
  journal={Annual Review of Economics},
  volume={13},
  number={1},
  pages={171--202},
  year={2021},
  publisher={Annual Reviews}
}

@article{li2025properties,
  title={Properties of CoVaR based on tail expansions of copulas},
  author={Li, Xiaoting and Joe, Harry},
  journal={Journal of Multivariate Analysis},
  pages={105510},
  year={2025},
  publisher={Elsevier}
}

@article{mainik2014dependence,
  title={On dependence consistency of CoVaR and some other systemic risk measures},
  author={Mainik, Georg and Schaanning, Eric},
  journal={Statistics \& Risk Modeling},
  volume={31},
  number={1},
  pages={49--77},
  year={2014},
  publisher={De Gruyter}
}

@article{engle2002dynamic,
  title={Dynamic conditional correlation: A simple class of multivariate generalized autoregressive conditional heteroskedasticity models},
  author={Engle, Robert},
  journal={Journal of Business \& Economic statistics},
  volume={20},
  number={3},
  pages={339--350},
  year={2002},
  publisher={Taylor \& Francis}
}

@article{huang2024monte,
  title={Monte Carlo Estimation of CoVaR},
  author={Huang, Weihuan and Lin, Nifei and Hong, L Jeff},
  journal={Operations Research},
  volume={72},
  number={6},
  pages={2337--2357},
  year={2024},
  publisher={INFORMS}
}

@article{wang2023tail,
title={Tail Gini Functional under Asymptotic Independence},
author={Wang, Zhaowen and Chen, Liujun and Li, Deyuan},
journal={Statistica Sinica},
doi = {10.5705/ss.202023.0426},
  note = {forthcoming},
year={2026}
}

@article{bauwens2006multivariate,
  title={Multivariate GARCH models: a survey},
  author={Bauwens, Luc and Laurent, S{\'e}bastien and Rombouts, Jeroen VK},
  journal={Journal of Applied Econometrics},
  volume={21},
  number={1},
  pages={79--109},
  year={2006},
  publisher={Wiley Online Library}
}

@article{ledford1997modelling,
  title={Modelling dependence within joint tail regions},
  author={Ledford, Anthony W and Tawn, Jonathan A},
  journal={Journal of the Royal Statistical Society: Series B (Statistical Methodology)},
  volume={59},
  number={2},
  pages={475--499},
  year={1997},
  publisher={Wiley Online Library}
}

@article{nolde2020conditional,
  title={Conditional extremes in asymmetric financial markets},
  author={Nolde, Natalia and Zhang, Jinyuan},
  journal={Journal of Business \& Economic Statistics},
  volume={38},
  number={1},
  pages={201--213},
  year={2020},
  publisher={Taylor \& Francis}
}

@article{leng2024asymptotics,
  title={Asymptotics of CoVaR Inference In Two-Quantile-Regression},
  author={Leng, Xuan and He, Yi and Hou, Yanxi and Peng, Liang},
  journal={Available at SSRN 4816475},
  year={2024}
}

@article{girardi2013systemic,
  title={Systemic risk measurement: Multivariate GARCH estimation of CoVaR},
  author={Girardi, Giulio and Erg{\"u}n, A Tolga},
  journal={Journal of Banking \& Finance},
  volume={37},
  number={8},
  pages={3169--3180},
  year={2013},
  publisher={Elsevier}
}

@article{oh2018time,
  title={Time-varying systemic risk: Evidence from a dynamic copula model of cds spreads},
  author={Oh, Dong Hwan and Patton, Andrew J},
  journal={Journal of Business \& Economic Statistics},
  volume={36},
  number={2},
  pages={181--195},
  year={2018},
  publisher={Taylor \& Francis}
}

@article{bianchi2023non,
  title={Non-Gaussian models for CoVaR estimation},
  author={Bianchi, Michele Leonardo and De Luca, Giovanni and Rivieccio, Giorgia},
  journal={International Journal of Forecasting},
  volume={39},
  number={1},
  pages={391--404},
  year={2023},
  publisher={Elsevier}
}

@article{hill1975simple,
  title={A simple general approach to inference about the tail of a distribution},
  author={Hill, Bruce M},
  journal={The annals of statistics},
  pages={1163--1174},
  year={1975},
  publisher={JSTOR}
}

@article{einmahl2012m,
  author = {John H. J. Einmahl and Andrea Krajina and Johan Segers},
title = {{An M-estimator for tail dependence in arbitrary dimensions}},
volume = {40},
journal = {Annals of Statistics},
number = {3},
publisher = {Institute of Mathematical Statistics},
pages = {1764 -- 1793},
year = {2012},
doi = {10.1214/12-AOS1023},
URL = {https://doi.org/10.1214/12-AOS1023}
}

@article{lehtomaa2020asymptotic,
  title={Asymptotic independence and support detection techniques for heavy-tailed multivariate data},
  author={Lehtomaa, Jaakko and Resnick, Sidney I},
  journal={Insurance: Mathematics and Economics},
  volume={93},
  pages={262--277},
  year={2020},
  publisher={Elsevier}
}

@article{lalancette2021rank,
  title={Rank-based estimation under asymptotic dependence and independence, with applications to spatial extremes},
  author={Lalancette, Micha{\"e}l and Engelke, Sebastian and Volgushev, Stanislav},
  journal={Annals of Statistics},
  volume={49},
  number={5},
  pages={2552--2576},
  year={2021},
  publisher={Institute of Mathematical Statistics}
}

@article{ledford1996statistics,
  title={Statistics for near independence in multivariate extreme values},
  author={Ledford, Anthony W and Tawn, Jonathan A},
  journal={Biometrika},
  volume={83},
  number={1},
  pages={169--187},
  year={1996},
  publisher={Oxford University Press}
}

@article{nolde2022extreme,
  title={An extreme value approach to CoVaR estimation},
  author={Nolde, Natalia and Zhou, Chen and Zhou, Menglin},
  journal={arXiv preprint arXiv:2201.00892},
  year={2022}
}

@book{de2007extreme,
  title={Extreme Value Theory: An Introduction},
  author={de Haan, Laurens and Ferreira, Ana},
  year={2006},
  publisher={Springer}
}

@article{cai2020estimation,
  title={Estimation of the marginal expected shortfall under asymptotic independence},
  author={Cai, Juan-Juan and Musta, Eni},
  journal={Scandinavian Journal of Statistics},
  volume={47},
  number={1},
  pages={56--83},
  year={2020},
  publisher={Wiley Online Library}
}

@article{zhang2017random,
  title={Random threshold driven tail dependence measures with application to precipitation data analysis},
  author={Zhang, Zhengjun and Zhang, Chunming and Cui, Qiurong},
  journal={Statistica Sinica},
  volume={27},
  number={2},
  pages={685--709},
  year={2017},
  publisher={JSTOR}
}

@techreport{adrian2011covar,
  title={CoVaR},
  author={Adrian, Tobias and Brunnermeier, Markus K},
  year={2011},
  institution={National Bureau of Economic Research}
}

@article{acharya2017measuring,
    author = {Acharya, Viral V. and Pedersen, Lasse H. and Philippon, Thomas and Richardson, Matthew},
    title = "{Measuring Systemic Risk}",
    journal = {Review of Financial Studies},
    volume = {30},
    number = {1},
    pages = {2-47},
    year = {2017},
    month = {10}
}

@article{le2018dependence,
  title={Dependence properties of spatial rainfall extremes and areal reduction factors},
  author={Le, Phuong Dong and Davison, Anthony C and Engelke, Sebastian and Leonard, Michael and Westra, Seth},
  journal={Journal of Hydrology},
  volume={565},
  pages={711--719},
  year={2018},
  publisher={Elsevier}
}

@article{wadsworth2012dependence,
  title={Dependence modelling for spatial extremes},
  author={Wadsworth, Jennifer L and Tawn, Jonathan A},
  journal={Biometrika},
  volume={99},
  number={2},
  pages={253--272},
  year={2012},
  publisher={Oxford University Press}
}

@article{das2018risk,
  title={Risk contagion under regular variation and asymptotic tail independence},
  author={Das, Bikramjit and Fasen-Hartmann, Vicky},
  journal={Journal of Multivariate Analysis},
  volume={165},
  pages={194--215},
  year={2018},
  publisher={Elsevier}
}

@article{kulik2015heavy,
  title={Heavy tailed time series with extremal independence},
  author={Kulik, Rafa{\l} and Soulier, Philippe},
  journal={Extremes},
  volume={18},
  pages={273--299},
  year={2015},
  publisher={Springer}
}

@article{sun2022extreme,
  title={Extreme Behaviors of the Tail Gini-Type Variability Measures},
  author={Sun, Hongfang and Chen, Yu},
  journal={Probability in the Engineering and Informational Sciences},
  volume={37},
  number={4},
  pages={928--942},
  year={2023},
  publisher={Cambridge University Press}
}
\end{document}